\documentclass[12pt]{article}

\usepackage{amssymb,amsthm,amscd,array}

\let\ssection=\section
\renewcommand{\section}{\setcounter{equation}{0}\ssection}

\newcommand{\bbR}{\mathbb{R}}

\newcommand{\bbC}{\mathbb{C}}

\newcommand{\bone}{{\bf 1}}

\newcommand{\cA}{{\mathcal{A}}}
\newcommand{\cB}{{\mathcal{B}}}
\newcommand{\tB}{{\widetilde{B}}}

\newcommand{\const}{\mathrm{const}}

\newcommand{\cD}{{\mathcal{D}}}

\newcommand{\cH}{\mathcal{H}}
\newcommand{\cM}{{\mathcal{M}}}

\newcommand{\cE}{{\mathcal{E}}}
\newcommand{\Diff}{\mathrm{Diff}}

\newcommand{\Div}{\mathrm{Div}}

\newcommand{\cF}{{\mathcal{F}}}

\newcommand{\rg}{\mathrm{g}}
\newcommand{\brg}{\overline{\mathrm{g}}}
\newcommand{\hrg}{\widehat{\mathrm{g}}}
\newcommand{\trg}{\widetilde{\mathrm{g}}}

\newcommand{\ii}{{\mathsf{i}}}

\newcommand{\Pol}{\mathrm{Pol}}
\newcommand{\cQ}{{\mathcal{Q}}}

\newcommand{\dq}{\dot{q}}
\newcommand{\ddq}{\ddot{q}}

\newcommand{\cS}{{\mathcal{S}}}

\newcommand{\SO}{\mathrm{SO}}

\newcommand{\so}{\mathrm{o}}
\newcommand{\homega}{\widehat{\omega}}
\newcommand{\hOmega}{\widehat{\Omega}}

\newcommand{\Tr}{\mathrm{Tr}}
\newcommand{\trP}{\Tr(P)}

\newcommand{\Vect}{\mathrm{Vect}}
\newcommand{\tu}{\tilde{u}}
\newcommand{\tv}{\tilde{v}}
\newcommand{\vol}{\mathrm{vol}}

\newcommand{\half}{{\frac{1}{2}}}

\newcommand{\la}{{\langle}}
\newcommand{\ra}{{\rangle}}

\newcommand{\beq}{\begin{equation}}  
\newcommand{\eeq}{\end{equation}}

\begin{document}

\baselineskip=16pt

\def\a{\alpha}
\def\b{\beta}
\def\c{\gamma}
\def\d{\delta}
\def\g{\gamma}
\def\om{\omega}
\def\r{\rho}
\def\s{\sigma}
\def\vfi{\varphi}
\def\l{\lambda}
\def\m{\mu}
\def\implies{\Rightarrow}

\oddsidemargin .1truein
\newtheorem{thm}{Theorem}[section]
\newtheorem{lem}[thm]{Lemma}
\newtheorem{cor}[thm]{Corollary}
\newtheorem{pro}[thm]{Proposition}
\newtheorem{ex}[thm]{Example}
\newtheorem{rmk}[thm]{Remark}
\newtheorem{defi}[thm]{Definition}

\title{A new integrable system on the sphere\\ 
and\\ 
conformally equivariant quantization}

\author{C. DUVAL\footnote{mailto:duval@cpt.univ-mrs.fr}\\
Centre de Physique Th\'eorique, CNRS, 
Luminy, Case 907\\ 
F-13288 Marseille Cedex 9 (France)\footnote{ 
UMR 6207 du CNRS associ\'ee aux 
Universit\'es d'Aix-Marseille I et II 
et Universit\'e du Sud Toulon-Var; 
Laboratoire 
affili\'e \`a la FRUMAM-FR2291}
\and
G. VALENT\footnote{mailto:valent@lpthe.jussieu.fr}\\
Laboratoire de Physique Th\'eorique et des Hautes Energies,\\
2, Place Jussieu\\
F-75251 Paris Cedex 5 (France)\footnote{
UMR 7589 du CNRS associ\'ee \`a l'Universit\'e Paris VI}
}

\date{14 September 2010}

\maketitle

\thispagestyle{empty}

\begin{abstract}
Taking full advantage of two independent projectively equivalent metrics on the ellipsoid leading to Liouville integrability of the geo\-desic flow via the well-known Jacobi-Moser system, we disclose a novel integrable system on the sphere $S^n$, namely the \textit{dual} Moser system. The latter falls, along with the Jacobi-Moser and Neumann-Uhlenbeck systems, into the category of (locally) St\"ackel systems. Moreover, it is proved that quantum integrability of both Neumann-Uhlenbeck and dual Moser systems is insured by means of the conformally equivariant quantization procedure.
\end{abstract}

\bigskip
\noindent
Preprint: CPT-P051-2010

\bigskip
\noindent
\textbf{Keywords:} Classical integrability, projectively equivalent metrics, St\"ackel sys\-tems, conformally equivariant quantization, quantum integrability.

\newpage
\tableofcontents

\section{Introduction}\label{Intro}

In the wake of the celebrated results of Moser \cite{Mos} concerning the classical integrability of the geodesic flow on the ellipsoid (first proved by Jacobi in the three-dimensional case), and of the Neumann-Uhlenbeck problem on the $n$-dimensional spheres, we will present a (to our knowledge) new integrable system which relies on a preferred confor\-mally flat metric on $S^n$. This integrable system is actually dual (in the sense of projective equivalence  \cite{Tab1,MT}) to the above-mentioned Jacobi-Moser system.

Among the techniques found in the literature, we will mention two different constructs that produce integrable systems together with their Poisson-commuting first integrals:
\begin{itemize}
\item
Bihamiltonian systems initiated by Magri \cite{Mag} and Benenti \cite{Ben}, and further elaborated by, e.g., Ibort, Magri, and Marmo \cite{IMM}, 
Falqui and Pedroni \cite{FP}. 
\item
Projectively equivalent systems discovered by Levi-Civita \cite{LevCiv}, and geometrical\-ly developed by Tabachnikov \cite{Tab1,Tab2}, Topalov and Matveev \cite{TM}. 
\end{itemize}
The equivalence of these two theories has been proved by Bolsinov and Matveev \cite{BM}.

We have chosen to work using Tabachnikov's approach. With the help of his general construction  \cite{Tab1,OT} providing, e.g., the Jacobi-Moser first integrals, and, adopting a ``dual'' approach, we derive the Poisson-commuting first integrals of the new geodesic flow, which we call the dual Moser system. This enables us to provide a \textit{global} expression for these new first integrals.

Our dual Moser system turns out to be locally St\"ackel (in ellipsoidal coordinates); it shows up as a ``mirror image'' of Jacobi-Moser relatively to Neumann-Uhlenbeck (see Table \ref{Table}).

In contradistinction to the Jacobi-Moser system, conformal flatness of the new system is a fundamental input at the classical, and at the quantum level as well.

To deal with quantum integrability of these systems, we will resort to its commonly accepted definition, namely that the quantized first integrals should still be in involution. This, of course, leaves open the choice of an adapted quantization procedure.

It has been shown \cite{DV} that St\"ackel systems, using Carter's quantization pres\-cription \cite{Car}, do remain integrable at the quantum level provided Robertson's condition holds \cite{Rob}. It is therefore worthwhile to study quantum integrabi\-lity making use of a genuine quantization theory that takes into account the conformal geometry underlying the dual Moser system: the conformally equivariant quantization \cite{DLO,DO}. One striking discovery is that the dual Moser system passes the quantum test using the conformally equivariant quantiza\-tion. We note that the same is true for the Neumann-Uhlenbeck system.

\subsection{Prolegomena}

Let us recall that two \textit{independent} metrics on a Riemannian manifold are said to be projectively equivalent if they have the same \textit{unparametrized} geodesics. A shown in, e.g., \cite{TM,BM} this entails that the associated geodesic flows are Liouville-integrable. We will resort to this pathbreaking result in the specific, historical, yet fundamental example of the geodesic flow of the ellipsoid.

The key point of our approach to Liouville/quantum integrability of the geodesic flow of the ellipsoid, $\cE^n$, lies in the fact that $\cE^n=\{Q\in\bbR^{n+1}\big\vert{}\sum_{\alpha=0}^{n}{(Q^{\alpha})^2/a_{\alpha}}=1\}$ with metric $\rg_1=\sum_{\alpha=0}^{n}{(dQ^{\alpha})^2}\big\vert_{\cE^n}$ admits, as a matter of fact, another independent, and projectively equivalent Riemannian metric, $\rg_2$, that we will introduce shortly. 

Put $Q^{\alpha}=q^{\alpha}\sqrt{a_{\alpha}}$ for all $\alpha=0,\ldots,n$, so that we have $q\cdot{}q=\sum_{\alpha=0}^{n}{(q^{\alpha})^2}=1$. The mapping $Q\mapsto{}q:\cE^n\to{}S^n$ is a diffeomorphism and the metric on~$S^n$, induced from the Euclidean ambient metric, reads now\footnote{To avoid clutter, we will oftentimes write $q_\alpha\equiv{}q^\alpha$, as no confusion occurs in Euclidean space.} 
\begin{equation}
\rg_1=\rg\big\vert_{S^n}
\qquad
\mathrm{where}
\qquad
\rg=\sum_{\alpha=0}^{n}{a_{\alpha} dq_{\alpha}^2}
\label{g1}
\end{equation}
where $\rg$ is, at the moment, viewed as a (flat) metric on $\bbR^{n+1}$.
The equations of the geodesics of the ellipsoid are well-known, and retain the form
\begin{equation}
\ddq^{\alpha}+\Gamma^{\alpha}_{\beta\gamma}\dq^{\beta}\dq^{\gamma}=0,
\label{geod1}
\end{equation}
for all $\alpha=0,\ldots,n$, where the Christoffel symbols of $(S^n,\rg_1)$ are given by 
\begin{equation}
\Gamma^{\alpha}_{\beta\gamma}=\frac{q^{\alpha}}{a_{\alpha}}\frac{\delta_{\beta\gamma}}{\sum{q_\lambda^2/a_{\lambda}}}
\label{Christoffel}
\end{equation}
for all $\alpha,\beta,\gamma=0,\ldots,n$. (See also Equation (\ref{XL1}) yielding the associated geodesic spray.) We note that the constraint equation 
\begin{equation}
\dq^2+q\cdot\ddq=0
\label{theConstraint}
\end{equation}
is indeed satisfied by (\ref{geod1}).

Following \cite{Tab1,TM}, let us define, on $S^n$, the conformally flat metric
\begin{equation}
\rg_2=\brg\big\vert_{S^n}
\qquad
\mathrm{with}
\qquad
\brg=\frac{1}{\sum{q_{\beta}^2/a_{\beta}}}\sum_{\alpha=0}^{n}{dq_{\alpha}^2} 
\label{g2}
\end{equation}
where the conformally flat metric $\brg$ is defined on $\bbR^{n+1}\!\setminus\!\{0\}$.

\goodbreak

The equations of the geodesics for the latter metric are readily found by using the expression of the Christoffel symbols of $\brg$, namely
\begin{equation}
\overline{\Gamma}_{\beta\gamma}^{\alpha}
=
\frac{1}{\sum{q_{\lambda}^2/a_{\lambda}}}\left[
\delta_{\beta\gamma}\frac{q^{\alpha}}{a_{\alpha}}
-\delta_{\beta}^{\alpha}\frac{q_{\gamma}}{a_{\gamma}}
-\delta_{\gamma}^{\alpha}\frac{q_{\beta}}{a_{\beta}}\right]
\label{Gammabar}
\end{equation}
for all $\alpha,\beta,\gamma=0\ldots,n$. One obtains the equations of the geodesic of $(\bbR^{n+1}\!\setminus\!\{0\},\brg)$, viz.,
\begin{equation}
\ddq^{\alpha}+\frac{q^{\alpha}}{a_{\alpha}}\frac{\sum{\dq_{\beta}^2}}{\sum{q_{\gamma}^2/a_{\gamma}}}=
2\dq^{\alpha}\,\frac{\sum{q_\beta\dq^{\beta}/a_{\beta}}}{\sum{q_{\gamma}^2/a_{\gamma}}}
\label{geod2}
\end{equation}
for all $\alpha=0,\ldots,n$. The latter, suitably restrained to $S^n$, is precisely Equation~(\ref{geod1}) with a different parametrization (again, the constraint (\ref{theConstraint}) is duly preserved by Equation (\ref{geod2})); the metrics $\rg_1$ and $\rg_2$ are projectively equivalent. 

We will put this fact in broader perspective within Section \ref{NewMoser}.

\begin{rmk}
{\rm
The $\Gamma_{\beta\gamma}^{\alpha}$ and $\overline{\Gamma}_{\beta\gamma}^{\alpha}$, given by (\ref{Christoffel}) and (\ref{Gammabar}) respectively, may be viewed as the components of two projectively equivalent linear connections $\nabla$ and $\overline{\nabla}$ on $\bbR^{n+1}\!\setminus\!\{0\}$. While $\overline{\nabla}$ is clearly the Levi-Civita connection of the metric $\brg$, the connection $\nabla$ is, instead, a \textit{Newton-Cartan} connection (see, e.g., \cite{Kun}). This means that $\nabla$ is a symmetric linear connection that parallel-transports a (spacelike) contravariant symmetric $2$-tensor $\gamma=\sum_{\alpha,\beta=0}^n{\gamma^{\alpha\beta}\partial_{q^\alpha}\otimes\partial_{q^\beta}}$ and a (timelike) $1$-form $\theta=\sum_{\alpha=0}^n{\theta_\alpha{}dq^\alpha}$ spanning $\ker(\gamma)$. Here, the degenerate ``metric'' is given by
\begin{equation}
\gamma^{\alpha\beta}
=
\frac{1}{a_\alpha}\delta^{\alpha\beta}
-\frac{q^\alpha{}q^\beta}{a_\alpha{}a_\beta}\frac{1}{\sum{q_{\lambda}^2/a_{\lambda}}}
\qquad
\hbox{and}
\qquad
\theta_\alpha=q_\alpha.
\label{NC}
\end{equation}
}
\end{rmk}

\subsection{Main results}

The main results of our article can be summarized as follows.

\begin{thm}
The geodesic flow on $(T^*S^n,\sum_\alpha{dp_\alpha\wedge{}dq^\alpha})$ above the conformally flat manifold $(S^n,\rg_2)$ is Liouville-integrable, and admits the following set of Poisson-commuting first integrals
\begin{equation}
F_\alpha
=
q_\alpha^2\sum_{\beta=0}^n{a_\beta{}p_\beta^2}
+
\sum_{\beta\neq\alpha}{\frac{(a_\alpha{}p_\alpha{}q_\beta-a_\beta{}p_\beta{}q_\alpha)^2}{a_\alpha-a_\beta}}
\label{FalphaMainResult}
\end{equation}
with $\alpha=0,\ldots,n$. We will call the system $(F_0,\ldots,F_n)$ the dual Moser system.
\end{thm}
This theorem follows directly from Propositions \ref{PoissonCommutingFalpha}, \ref{ProPBFalphaFbeta}, and \ref{ProDiracBracket}.

\goodbreak

Let us call $\cD_{\half,\half}(S^n)$ the space of differential operators on $S^n$ with arguments and values in the space of $\half$-densities of the sphere $S^n$; we, likewise, denote by $\Pol(T^*S^n)$ the space of fiberwise polynomial functions on $T^*M$. It has been proved \cite{DLO} that there exists a unique invertible linear mapping $\cQ_{\half,\half}:\Pol(T^*S^n)\to\cD_{\half,\half}(S^n)$ that (i) intertwines the action of the conformal group $\mathrm{O}(n+1,1)$ and (ii) preserves the principal symbol: we call it the \textit{conformally equivariant quantization} mapping.
\begin{thm}
Quantum integrability of the dual Moser system holds true in terms of the conformally equivariant quantization $\cQ_{\half,\half}$.
\end{thm}

This last result stems from Theorem \ref{ThmDO}, Propositions \ref{DualMoserConfEquivQuantInt}, and \ref{ProConfQ}.

\subsection{Plan of the article}
The paper is organized as follows.

Section \ref{NewMoser} gives us the opportunity to introduce three distinguished classically integrable systems on the sphere, namely, the Jacobi-Moser system, its dual counterpart, and the Neumann-Uhlenbeck system. The construction of the set of mutual\-ly Poisson-commuting first integrals is reviewed and specialized to the case of the dual Moser system. The resulting system is shown to be St\"ackel, the ellipsoidal co\-ordinates being the separating ones.

In Section \ref{QuantInt} we address the quantum integrability issue of these systems, in terms of the conformally equivariant quantization. We prove that the Neumann-Uhlenbeck and the dual Moser systems are, indeed, quantum integrable using this quantization method, well adapted to the conformal flatness of configuration space.

Section \ref{Conclusion}, provides a conclusion to the present article, and gathers some perspectives for future work.

\section{A novel integrable system on the sphere: the dual Moser system}
\label{NewMoser}

\subsection{Projectively equivalent metrics and conservation laws}\label{biHam}

Let us recall, almost \textit{verbatim}, Tabachnikov's construction \cite{Tab1} of a maximal set of independent Poisson-commuting first integrals for a special Liouville-integrable system, namely a bi-Hamiltonian system associated with two projectively equivalent metrics, $\rg_1$ and $\rg_2$, on a configuration manifold $M$.\footnote{The formalism can be easily extended to the case of Finsler structures \cite{Tab1}; here, we will not need such a generality.} See also \cite{MT,BM} for an alternative construction.

\goodbreak

We start with two Riemannian manifolds $(M,\rg_1)$ and $(M,\rg_2)$ of dimension~$n$. 
The tangent bundle $TM$ is endowed with two distinguished $1$-forms $\lambda_1$ and $\lambda_2$, namely $\lambda_N=\rg^*_N\theta$, where~$\theta$ is the canonical $1$-form of $T^*M$, and $\rg_N:TM\to{}T^*M$ is viewed as a bundle isomorphism. We then write, locally, $\lambda_N=\rg^N_{ij}u^idx^j$ where $\rg_N=\rg^N_{ij}(x)dx^i\otimes{}dx^j$ for $N=1,2$. Likewise, the Lagrangian functions to consider are the fiberwise quadratic polynomials $L_N=\half\rg^N_{ij}(x)u^iu^j$. Denote by $\omega_N=d\lambda_N$ the corresponding symplectic $2$-forms of $TM$, and also by $X_N=X_{L_N}$ the associated geodesic sprays. We have 
\begin{equation}
\lambda_N(X_N)=2L_N
\qquad
\hbox{and}
\qquad
\omega_N(X_N)=-dL_N
\label{lambdaandomega}
\end{equation}
for all $N=1,2$. The (maximal) integral curves of the vector fields $X_N$ on~$TM$ project onto configuration space as the geodesics of $(M,\rg_N)$.

\goodbreak

Introduce then the diffeomorphism $\phi:TM\to{}TM$ defined by $\phi(x,u)=(x,\tu)$ where $\tu=u\sqrt{L_1(x,u)/L_2(x,u)}$.\footnote{Note that $\phi$ is the identity on the zero section of $TM$.} This diffeomorphism is, indeed, designed to relate the two Lagrangians, viz.,
\begin{equation}
L_1=\phi^*L_2.
\label{H1=phistarH2}
\end{equation}
Clearly, the two metrics $\rg_1$ and $\rg_2$ have the same unparametrized geodesics (we write~$\rg_1\sim\rg_2$) iff
\begin{equation}
\phi_*(X_1)\wedge{}X_2=0
\label{phistarX1wedgeX2=0}
\end{equation}
i.e., iff the the push-forward $\phi_*(X_1)$ and $X_2$ are functionally dependent.

The method, to obtain a generating function for the conserved quantities in involution, consists then in singling out, apart from $\omega_1$, a preferred $X_1$-invariant $2$-form constructed in terms of $\omega_2$. 
\begin{pro}\label{ProLX1phistaromegaprime2}
Suppose that $\rg_1\sim\rg_2$, and define $\omega'_2=d(L_2^{-\half}\lambda_2)$, then 
\begin{equation}
L_{X_1}(\phi^*\omega'_2)=0.
\label{LX1omegaprime2=0}
\end{equation}
\end{pro}
\begin{proof}
We have $L_{X_1}(\phi^*\omega'_2)=d((\phi^*\omega'_2)(X_1))=\phi^*d(\omega'_2(\phi_*X_1))=\phi^*d(h\,\omega'_2(X_2))$, for some function $h$ (see (\ref{phistarX1wedgeX2=0})). The definition of $\omega'_2$ then readily yields $L_{X_1}(\phi^*\omega'_2)=\phi^*d(h\,d(L_2^{-\half}\lambda_2)(X_2))=\phi^*d(h(-\half{}L_2^{-\frac{3}{2}}(dL_2\wedge\lambda_2)(X_2)+L_2^{-\half}\omega_2(X_2)))=0$ in view of~(\ref{lambdaandomega}).
\end{proof}
So, the sought extra $X_1$-invariant $2$-form is $\phi^*\omega'_2$. (Note that $L_{X_1}(\phi^*\omega_2)\neq0$.) It enters naturally into the definition of a ``generating function'' $f_t\in{}C^\infty(TM,\bbR)$ of first-integrals given below.
\begin{cor} 
The function
\begin{equation}
f_t=\frac{(t^{-1}\omega_1+\phi^*\omega'_2)^n}{\omega_1^n}
\label{ft}
\end{equation}
is $X_1$-invariant whenever $t\neq0$.
\end{cor}

\subsection{The example of two projectively equivalent geodesic sprays for the $n$-sphere}

We recall the construction of the first-integrals in involution yielding the Liouville-integrability of the geodesic flow on $T\cE^n\cong{}TS^n$. 

Let us parametrize $T\bbR^{n+1}$ by the couples $q,v\in\bbR^{n+1}$. The constraints defining the embedding $TS^n\hookrightarrow{}T\bbR^{n+1}$ are
\begin{equation}
q^2:=\sum_{\alpha=0}^n{q_\alpha^2}=1
\qquad
\hbox{and}
\qquad
v\cdot{q}:=\sum_{\alpha=0}^n{v_\alpha{}q_\alpha}=0.
\label{ConstraintsTSn} 
\end{equation}

As already mentioned in Section \ref{Intro},
the (unparametrized) geodesics of the ellipsoid $\cE^n$
with semi-axes\footnote{We will, later on, deal with the choice $0<a_0<a_1<\cdots<a_n$.} $a_0, a_1,\cdots,a_n$ are precisely given by those of the sphere $S^n=\{q\in\bbR^{n+1}\vert\sum_{\alpha=0}^n{q_\alpha^2}=1\}$ endowed with either projectively equivalent metrics
\begin{equation}
\rg_1=\sum_{\alpha=0}^n{a_\alpha{}dq_\alpha^2}\Big\vert_{S^n}
\qquad
\&
\qquad
\rg_2=\frac{1}{B}\sum_{\alpha=0}^n{dq_\alpha^2}\Big\vert_{S^n}
\label{g1g2}
\end{equation}
where
\begin{equation}
B=\sum_{\alpha=0}^n{\frac{q_\alpha^2}{a_\alpha}}.
\label{B}
\end{equation}

The corresponding Lagrangians on $TS^n$ are respectively
\begin{equation}
L_1=\half A
\qquad
\&
\qquad
L_2=\frac{1}{2B}\sum_{\alpha=0}^n{v_\alpha^2}
\label{L1L2}
\end{equation}
where
\begin{equation}
A=\sum_{\alpha=0}^n{a_\alpha{}v_\alpha^2}.
\label{A} 
\end{equation}
The associated Cartan $1$-forms then read in this case
\begin{equation}
\lambda_1=\sum_{\alpha=0}^n{a_\alpha{}v_\alpha{}dq_\alpha}
\qquad
\&
\qquad
\lambda_2=\frac{1}{B}\sum_{\alpha=0}^n{v_\alpha{}dq_\alpha}.
\label{lambda1lambda2} 
\end{equation}

\goodbreak

\begin{pro}
(i) The geodesic sprays for the metrics $\rg_N$ are given by the Hamiltonian vector fields $X_N=X_{L_N}$, for $N=1,2$, namely
\begin{equation}
X_1
=
\sum_{\alpha=0}^n{v_\alpha\frac{\partial}{\partial{}q_\alpha}}
-
\frac{v^2}{B}\sum_{\alpha=0}^n{\frac{q_\alpha}{a_\alpha}\frac{\partial}{\partial{}v_\alpha}}
\label{XL1} 
\end{equation}
and
\begin{equation}
X_2
=
\sum_{\alpha=0}^n{v_\alpha\frac{\partial}{\partial{}q_\alpha}}
-
\frac{1}{B}\sum_{\alpha=0}^n{\Big(2v_\alpha\sum_{\beta=0}^n{\frac{v_\beta{}q_\beta}{a_\beta}}-v^2\frac{q_\alpha}{a_\alpha}\Big)\frac{\partial}{\partial{}v_\alpha}}
\label{XL2} 
\end{equation}
respectively.

(ii) Condition (\ref{phistarX1wedgeX2=0}) holds true, implying $\rg_1\sim\rg_2$.
\end{pro}
\begin{proof}
Using (\ref{lambdaandomega}) together with the constraints (\ref{ConstraintsTSn}), we  thus have to solve for $X_1$, resp. $X_2$, the equation $\omega_N(X_N)+dL_N+\lambda\,d(q^2-1)+\mu\,d(v\cdot{}q)=0$ where $\lambda$ and $\mu$ are Lagrange multipliers. The latter are, \textit{in fine}, completely determined and readily yield (\ref{XL1}), resp. (\ref{XL2}). 

Now, the diffeomorphism $\phi:(q,v)\mapsto(q,\tv)$ introduced in Section \ref{biHam} is given by $\tv=v\sqrt{AB/v^2}$; routine calculation yields $\phi_*(\partial_{q_\alpha})=\partial_{q_\alpha}+q_\alpha/(a_\alpha{}B)\,\cE$ with $\cE=\sum{v_\alpha\partial_{v_\alpha}}$ the Euler vector field; also $\phi_*(\partial_{v_\alpha})=\sqrt{AB/v^2}(\partial_{v_\alpha}-v_\alpha(v^{-2}-a_\alpha/A)\,\cE)$. This, along with the constraint $\sum{v_\alpha{}q_\alpha=0}$, helps us prove Equation~(\ref{phistarX1wedgeX2=0}).
\end{proof}

\subsection{The Jacobi-Moser system}

Let us review here, and in some detail, the main result obtained by Tabachnikov \cite{Tab1} via the general procedure of Section~\ref{biHam}, starting with the geodesic flow on $(S^n,\rg_1)$.


The diffeomorphism $\phi$ of $TS^n$, namely $\phi(q,v)=(q,\tv=v\sqrt{L_1/L_2})$, is such that $\tv=v\sqrt{AB/v^2}$ where $v^2=\sum{v_\alpha^2}$; whence
$
\phi^*\lambda_2=C\sqrt{A}\,\sum{v_\alpha{}dq_\alpha}
$,
where
\begin{equation}
C=\frac{1}{\sqrt{B v^2}}.
\label{C} 
\end{equation}
Now, since $\phi^*\omega'_2=d(L_1^{-\half}\phi^*\lambda_2)$, easy computation then leads to
\begin{equation}
\phi^*\omega'_2=C\sqrt{2}\,\Omega_1
\quad
\hbox{where}
\quad
\Omega_1=\sum_{\alpha=0}^n{dv_\alpha\wedge{}dq_\alpha}+\frac{dC}{C}\wedge\sum_{\alpha=0}^n{v_\alpha{}dq_\alpha}.
\label{phistaromegaprime2} 
\end{equation}

The function $C$ defined by (\ref{C}) is a first integral of the system, namely
\begin{equation}
X_1C=0.
\label{X1C=0} 
\end{equation}
This function $C$ is the \textit{Joachimsthal} first-integral of the geodesic flow of $(S^n,\rg_1)$.

The next step consists in lifting the $1$-parameter family of first-integrals (\ref{ft}) to $T\bbR^{n+1}$ by taking advantage of the constraints (\ref{ConstraintsTSn}), and to put
\begin{equation}
f_t=\frac{(t^{-1}\omega_1+\Omega_1)^n\wedge{}d(v.q)\wedge{}q\cdot{}dq}{\omega_1^n\wedge{}d(v.q)\wedge{}q\cdot{}dq}
\label{ftbis} 
\end{equation}
with a slight abuse of notation using the constancy of $C$ (see (\ref{X1C=0})) in  (\ref{phistaromegaprime2}).
Elementary calculation yields
\begin{equation}
f_t=
\frac{\left(\strut\omega_*^n-(n/v^2)\,\omega_*^{n-1}\wedge{}v\cdot{}dv\wedge{}v\cdot{}dq\right)\wedge{}q\cdot{}dv\wedge{}q\cdot{}dq}
{\omega_1^n\wedge{}(dv.q+v\cdot{}dq)\wedge{}q\cdot{}dq}
\label{ftter} 
\end{equation}
where $\omega_*=\sum{b_\alpha{}dv_\alpha\wedge{}dq_\alpha}-(1/v^2)\,v\cdot{}dv\wedge{}v\cdot{}dq$, together with $b_\alpha=t^{-1}a_\alpha+1$, for all $\alpha=0,\ldots,n$.
The following lemma \cite{Tab1} will be used to complete the calculation.
\begin{lem}\label{TechnicalLemma}
Let $\omega=\sum{c_\alpha{}dv_\alpha\wedge{}dq_\alpha}$ and $\omega_0=\sum{dv_\alpha\wedge{}dq_\alpha}$, then
\begin{eqnarray*}
\frac{\omega^n\wedge{}q\cdot{}dv\wedge{}q\cdot{}dq}{\omega_0^{n+1}}
&=&
n!\,\prod_{\alpha=0}^n{c_\alpha}\sum_{\alpha=0}^n{\frac{q_\alpha^2}{c_\alpha}}\\
\frac{\omega^{n-1}\wedge{}v\cdot{}dv\wedge{}v\cdot{}dq\wedge{}q\cdot{}dv\wedge{}q\cdot{}dq}{\omega_0^{n+1}}
&=&
(n-1)!\,\prod_{\alpha=0}^n{c_\alpha}\sum_{\alpha<\beta}{\frac{(v_\alpha{}q_\beta-v_ \beta{}q_ \alpha)^2}{c_\alpha{}c_\beta}}.
\end{eqnarray*}
\end{lem}
Using the two above formul\ae, we find $f_t=N/D$ where 
\begin{eqnarray*}
N&=&n!\prod{b_\alpha}\sum{\frac{q_\alpha^2}{b_\alpha}}-\frac{n(n-1)!}{v^2}\prod{b_\alpha}\sum_{\alpha<\beta}{\frac{(v_\alpha{}q_\beta-v_ \beta{}q_ \alpha)^2}{b_\alpha{}b_\beta}}\\
D&=&n!\prod{a_\alpha}\sum{q_\alpha^2/a_\alpha}.
\end{eqnarray*}
This entails $f_t=g_t/C$ (up to a constant overall factor), where $C$ is the \textit{Joachimsthal} first-integral, and
$$
g_t
=v^2\sum_{\alpha=0}^n{\frac{q_\alpha^2}{b_\alpha}}
-
\half\sum_{\alpha\neq\beta}{\frac{(v_\alpha{}q_\beta-v_ \beta{}q_ \alpha)^2}{b_\alpha{}b_\beta}}.
$$
Taking into account the expression $b_\alpha=t^{-1}a_\alpha+1$, and the constraints (\ref{ConstraintsTSn}), we end up with
\begin{equation}
g_t
=
\sum_{\alpha=0}^n{\frac{a_\alpha{}v_\alpha^2}{a_\alpha+t}}
-
\sum_{\alpha=0}^n{\frac{a_\alpha{}v_\alpha^2}{a_\alpha+t}}
\sum_{\alpha=0}^n{\frac{a_\alpha{}q_\alpha^2}{a_\alpha+t}}
+\left(\sum_{\alpha=0}^n{\frac{a_\alpha{}v_\alpha{}q_\alpha}{a_\alpha+t}}\right)^2.
\label{gt}
\end{equation}
At last, the first-integrals $F_\alpha$ defined by
\begin{equation}
g_t=\sum_{\alpha=0}^n{\frac{F_\alpha}{a_\alpha+t}}
\label{DefFalpha}
\end{equation}
are easily found to be
\begin{equation}
F_\alpha=a_\alpha{}v_\alpha^2+\sum_{\beta\neq\alpha}{\frac{a_\alpha{}a_\beta(v_\alpha{}q_\beta-v_ \beta{}q_ \alpha)^2}{a_\alpha-a_\beta}}.
\label{FalphaMoser}
\end{equation}
These are the \textit{Jacobi-Moser} first-integrals. In terms of the momenta $p_\alpha=a_\alpha{}v_\alpha$ (see~(\ref{lambda1lambda2})), they read
\begin{equation}
F_\alpha=\frac{p_\alpha^2}{a_\alpha}+\sum_{\beta\neq\alpha}{\frac{(p_\alpha{}a_\beta{}q_\beta-p_\beta{}a_\alpha{}q_ \alpha)^2}{a_\alpha{}a_\beta(a_\alpha-a_\beta)}}.
\label{FalphaBis}
\end{equation}

\begin{rmk}
{\rm
We indeed recover the Moser first integrals
\begin{equation}
F_\alpha=P_\alpha^2+\sum_{\beta\neq\alpha}{\frac{(P_\alpha{}Q_\beta-P_\beta{}Q_ \alpha)^2}{a_\alpha-a_\beta}}
\label{FalphaBisMoser}
\end{equation}
by means of the canonical transformation $(p,q)\mapsto(P,Q)$ where $P_\alpha=p_\alpha/\sqrt{a_\alpha}$, and $Q_\alpha=q_\alpha\sqrt{a_\alpha}$.
}
\end{rmk}

\subsection{The dual Moser system}

Let us now adopt a ``dual'' standpoint by exchanging the r\^ole of the two projectively equivalent metrics on the $n$-sphere, i.e., by letting $\rg_1\leftrightarrow\rg_2$ in the above derivation. In doing so, we will work out a complete set of commuting first-integrals of the geodesic flow on the conformally flat manifold $(S^n,\rg_2)$. This will turn out to provide a new integrable system on the $n$-sphere.

\subsubsection{The general construction}

Let us now apply the general procedure outlined in Section \ref{biHam} starting with the geodesic flow on $(S^n,\rg_2)$, and replacing \textit{mutatis mutandis} all reference to $\rg_1$ by that of $\rg_2$. 

We first need to work out the expression of the $2$-form $\phi^*\omega'_1$ in Proposition~\ref{ProLX1phistaromegaprime2}. The diffeomorphism $\phi$ of $TS^n$, viz., $\phi(q,v)=(q,\tv=v\sqrt{L_2/L_1})$, is $\tv=v\sqrt{v^2/(AB)}$ where, again, $v^2=\sum{v_\alpha^2}$. We hence get
$\phi^*\lambda_1=\sqrt{v^2/(AB)}\,\sum{a_\alpha{}v_\alpha{}dq_\alpha}$.
Now, since $\phi^*\omega'_1=d(L_2^{-\half}\phi^*\lambda_1)$, thanks to $\phi^*L_1=L_2$, we are led to
\begin{equation}
\phi^*\omega'_1=\sqrt{\frac{2}{A}}\,\hOmega_2
\qquad
\hbox{where}
\qquad
\hOmega_2=\sum_{\alpha=0}^n{a_\alpha{}dv_\alpha\wedge{}dq_\alpha}-\frac{dA}{2A}\wedge\lambda_1,
\label{phistaromegaprime1} 
\end{equation}
with $\lambda_1=\sum{a_\alpha{}v_\alpha{}dq_\alpha}$ (see (\ref{lambda1lambda2})).

We then find the \textit{new Joachimsthal} first-integral, $J$, of the geodesic flow of $(S^n,\rg_2)$ with the help of the expression (\ref{XL2}) of the geodesic spray $X_2$. In fact, easy calculation yields $X_2A=(2A/B)X_2(B)$, hence $X_2J=0$ where
\begin{equation}
J=\frac{A}{B^2}.
\label{J} 
\end{equation}
Utilizing this constant of the motion, and Equation (\ref{LX1omegaprime2=0}), we obtain, see (\ref{phistaromegaprime1}),
\begin{equation}
L_{X_2}\Omega_2=0
\qquad
\hbox{where}
\qquad
\Omega_2=\frac{\hOmega_2}{\!B}.
\label{LX2Omega2=0} 
\end{equation}

Again, and to ease the calculation, we will lift the $1$-parameter family of first-integrals (\ref{ft}) to $T\bbR^{n+1}$, using the constraints (\ref{ConstraintsTSn}), and put this time
\begin{equation}
f_t=\frac{(t\,\omega_2+\Omega_2)^n\wedge{}d(v\cdot{}q)\wedge{}q\cdot{}dq}{\omega_2^n\wedge{}d(v\cdot{}q)\wedge{}q\cdot{}dq}.
\label{ft2} 
\end{equation}
We trivially get
\begin{equation}
f_t
=
\frac{(t\,\homega_2+\hOmega_2)^n\wedge{}q\cdot{}dv\wedge{}q\cdot{}dq}{\homega_2^n\wedge{}q\cdot{}dv\wedge{}q\cdot{}dq}\\[6pt]
\label{ft2bis} 
\end{equation}
where $\homega_2=B\omega_2$, i.e.,
$
\homega_2
=
\sum{dv_\alpha\wedge{}dq_\alpha}-d\log B\wedge\sum{v_\alpha{}dq_\alpha}
$
and $\hOmega_2$ is as in (\ref{phistaromegaprime1}).

\goodbreak

Let us mention the following somewhat technical lemma.

\begin{lem}\label{SecondTechnicalLemma}
Upon defining $\homega_*=\sum{b_\alpha{}dv_\alpha\wedge{}dq_\alpha}$, where $b_\alpha=a_\alpha+t$, for every $\alpha=0,1,\ldots,n$, we have
\begin{eqnarray*}
(t\,\homega_2+\hOmega_2)^k
&=&
\homega_*^k
-k\,\homega_*^{k-1}\wedge\Big(t\,d\log B\wedge{}v\cdot{}dq+\half{}d\log A\wedge\lambda_1\Big)\\[6pt]
&&
+k(k-1)\,\homega_*^{k-2}\wedge{}t\,d\log B\wedge{}v\cdot{}dq\wedge\half{}d\log A\wedge\lambda_1
\end{eqnarray*}
for all $k=1,\ldots,n$, and
$$
\homega_*^{n-1}
=
(n-1)!\sum_{\alpha<\beta}{
\prod_{\gamma\neq\alpha,\beta}
{
\!\!b_\gamma{}dv_\gamma\wedge{}dq_\gamma
}
}
$$
\end{lem}
A somewhat demanding computation yields
\begin{equation}
f_t
=
\frac{\left(\strut\homega_*^n-n\,\homega_*^{n-1}\wedge{}dA/(2A)\wedge\lambda_1\right)\wedge{}q\cdot{}dv\wedge{}q\cdot{}dq}
{\homega_2^n\wedge{}q\cdot{}dv\wedge{}q\cdot{}dq}.
\label{ft2ter} 
\end{equation}
Resorting to Lemmas \ref{TechnicalLemma} and \ref{SecondTechnicalLemma} in order to evaluate $f_t$, we obtain the partial result
\begin{equation}
\homega_*^n\wedge{}q\cdot{}dv\wedge{}q\cdot{}dq=n!\prod_{\beta=0}^n{b_\beta}\left(\sum_{\alpha=0}^n {\frac{q_\alpha^2}{b_\alpha}}\right)\omega_0^{n+1}
\label{omegastarn} 
\end{equation}
where $\omega_0=\sum{dv_\alpha\wedge{}dq_\alpha}$. 
Likewise, some more effort is needed to find
\begin{equation}
\homega_*^{n-1}\wedge\half{}d\log{A}\wedge\lambda_1\wedge{}q\cdot{}dv\wedge{}q\cdot{}dq
=
(n-1)!
\prod_{\gamma=0}^n{b_\gamma}
\sum_{\alpha<\beta}{\frac{\left(a_\alpha{}v_\alpha{}q_\beta-a_ \beta{}v_ \beta{}q_ \alpha\right)^2}{b_\alpha{}b_\beta}
}\frac{\omega_0^{n+1}}{A}.
\label{omegastarn-1} 
\end{equation}
We just have to plug Equations (\ref{omegastarn}) and (\ref{omegastarn-1}) into the expression (\ref{ft2bis}) to find\footnote{We have $\homega_2^n\wedge{}q\cdot{}dv\wedge{}q\cdot{}dq=n!\,\omega_0^{n+1}$.}
\begin{equation}
f_t
=
{\displaystyle\prod_{\gamma=0}^n{b_\gamma}}
\left(
\sum_{\alpha=0}^n{\frac{q_\alpha^2}{b_\alpha}}
-\frac{1}{A}
\sum_{\alpha<\beta}{
\frac{\left(a_\alpha{}v_\alpha{}q_\beta-a_ \beta{}v_ \beta{}q_ \alpha\right)^2}{b_\alpha{}b_\beta}
}
\right).
\label{ft2quater} 
\end{equation}
We will, again, deal with the rescaled first-integral 
\begin{equation}
g_t=\frac{\!\!J}{\prod{b_\alpha}}\,f_t
\label{gt2} 
\end{equation}
where $J$ is an (\ref{J}), as a generating function of the sought conservative system $F_0,F_1,\ldots,F_n$. Using the definition $b_\alpha=a_\alpha+t$, and Equation (\ref{DefFalpha}), we readily prove that the geodesic flow on $TS^n$ above $(S^n,\rg_2)$ admits the following first integrals, viz.,
\begin{equation}
F_\alpha
=
\frac{1}{B^2}
\left(Aq_\alpha^2
+
\sum_{\beta\neq\alpha}{
\frac{\left(a_\alpha{}v_\alpha{}q_\beta-a_ \beta{}v_ \beta{}q_ \alpha\right)^2}{a_\alpha-a_\beta}
}
\right)
\label{ft2quattro} 
\end{equation}
with $\alpha=0,1,\ldots,n$, where $A$ and $B$ are given by Equations (\ref{A}) and (\ref{B}), respective\-ly.

With the help of the bundle isomorphism $\rg_2^{-1}:T^*S^n\to{}TS^n$ provided by the metric $\rg_2$, we can pull-back the previous first integrals to the cotangent bundle of~$S^n$. Whence the following result.
\begin{pro}\label{PoissonCommutingFalpha}
The geodesic flow on the cotangent bundle of $(S^n,\rg_2)$ admits the following set of first integrals, viz.,
\begin{equation}
F_\alpha
=
q_\alpha^2\sum_{\beta=0}^n{a_\beta{}p_\beta^2}
+
\sum_{\beta\neq\alpha}{\frac{(a_\alpha{}p_\alpha{}q_\beta-a_\beta{}p_\beta{}q_\alpha)^2}{a_\alpha-a_\beta}}
\label{Falpha}
\end{equation}
with $\alpha=0,\ldots,n$.
\end{pro}

In the next section we will prove that these first integrals are actually independent and are mutually Poisson commuting.

\subsubsection{Liouville integrability of the unconstrained system}

From now on, we choose to work in a purely Hamiltonian framework which will turn out to be well-suited to the quantization procedure that we will examine in the next section.

Let us recall that the Hamiltonian of the system is given by
\begin{equation}
H=\half\sum_{\alpha,\beta=0}^n{\rg_2^{\alpha\beta}p_\alpha{}p_\beta}=\half{}B\sum_{\alpha=0}^n{p_\alpha^2}
\label{newHam}
\end{equation}
where $B=\sum_{\alpha=0}^n{q_\alpha^2/a_\alpha}$ (see (\ref{B})).


We will show that the new set (\ref{Falpha}) of first-integrals of motion indeed turns the geodesic flow on the sphere $(S^n, \rg_2)$ into an integrable system dual to the Jacobi-Moser geodesic flow on the ellipsoid. 

\goodbreak

\begin{pro}\label{ProPBFalphaFbeta}
The functions $F_0,\ldots,F_n$ of $(T^*\bbR^{n+1},\sum_{\alpha=0}^n{dp_\alpha\wedge{}dq^\alpha})$ given by (\ref{Falpha}) are 
in involu\-tion, namely
\begin{equation}
\{F_\alpha,\,F_\beta\}=0
\label{PBFalphaFbeta}
\end{equation}
for all $\alpha,\beta=0,\ldots,n$. Moreover the following holds true
\begin{equation}
\{H,F_\alpha\}=
2\Big(
B\,a_\alpha{}p_\alpha^2
-
q_\alpha^2\sum_\beta{p_\beta^2}
\Big)
\sum_\gamma{p_\gamma{}q_\gamma}.
\label{PBHFalpha}
\end{equation}
\end{pro}

\begin{proof}
Write $F_\alpha=\cA_\alpha+\cB_\alpha$ with 
\begin{equation}
\cA_\alpha=q_\alpha^2{}J,
\qquad
J=\sum_{\beta=0}^n{a_\beta{}p_\beta^2},
\label{cAcB}
\end{equation}
where $J$ is the Joachimsthal first integral (\ref{J}), and
\begin{equation}
\cB_\alpha=\sum_{\beta\neq\alpha}\frac{\cM_{\alpha\beta}^2}{a_\alpha-a_\beta},
\qquad
\qquad
\cM_{\alpha\beta}=a_\alpha{}p_\alpha{}q_\beta-a_\beta{}p_\beta{}q_\alpha.
\label{cAcM}
\end{equation}

One can check the following relationships 
$$
\{\cA_\alpha,\cA_\beta\}=-4J{}q_\alpha{}q_\beta\cM_{\alpha\beta},
\qquad
\{\cA_\alpha,\cB_\beta\}=4J{}a_\alpha{}q_\alpha{}q_\beta\frac{\cM_{\alpha\beta}}{a_\alpha-a_\beta},
\qquad
\{\cB_\alpha,\cB_\beta\}=0,
$$
for all $\alpha,\beta=0,\ldots,n$, which readily imply Equation (\ref{PBFalphaFbeta}). 

Let us furthermore observe that we have the following Poisson brackets, viz.,
$$
\{H,\cA_\alpha\}=2J B q_\alpha{}p_\alpha-4q_\alpha^2\frac{H}{B}\sum_{\beta=0}^n{p_\beta{}q_\beta}
$$
and
$$
\{H,\cB_\alpha\}=-2J B q_\alpha{}p_\alpha+2Ba_\alpha{}p_\alpha^2\sum_{\beta=0}^n{p_\beta{}q_\beta},
$$
which proves Equation (\ref{PBHFalpha}). The proof is complete.
\end{proof}

\begin{rmk}
{\rm
Notice in contradistinction to the Jacobi-Moser case, that (i) the metric~$\brg$ given by (\ref{g2}) on the ambient space $\bbR^{n+1}\!\setminus\!\{0\}$ is no longer flat, and (ii) the conserva\-tion relations $\{H,F_\alpha\}=0$ for $\alpha=0,\ldots,n$ are only valid for the constrained system, see~(\ref{PBHFalpha}), where $p\cdot{}q=0$. 
}
\end{rmk}

Let us also mention the interesting relations
\begin{eqnarray}
\label{Rel1}
\sum_{\alpha=0}^n{F_\alpha}&=&\sum_{\alpha=0}^n{q_\alpha^2}\sum_{\beta=0}^n{a_\beta{}p_\beta^2}\\
\label{Rel2}
\sum_{\alpha=0}^n{\frac{F_\alpha}{a_\alpha}}&=&\left(\sum_{\alpha=0}^n{p_\alpha{}q_\alpha}\right)^2\\
\label{Rel3}
\sum_{\alpha=0}^n{\frac{F_\alpha}{a_\alpha^2}}&=&-2H+2\sum_{\alpha=0}^n{p_\alpha{}q_\alpha}\sum_{\beta=0}^n{\frac{p_\beta{}q_\beta}{a_\beta}},
\end{eqnarray}
of which the last one leads to another proof of (\ref{PBHFalpha}).

\goodbreak

\subsubsection{The Dirac brackets}

Our goal is now to deduce from the knowledge of (\ref{Falpha}) 
independent quantities in involution $I_1,\ldots,I_n$ on 
$(T^*S^n,\sum_{i=1}^n{d\xi_i\wedge{}dx^i})$ from the symplectic embedding 
$
\iota:T^*S^n\hookrightarrow{}T^*\bbR^{n+1}
$ 
defined by the constraints
\begin{equation}
Z_1(p,q)=\sum_{\alpha=0}^n{q_\alpha^2}-1=0,
\qquad
Z_2(p,q)=\sum_{\alpha=0}^n{p_\alpha{}q_\alpha}=0.
\label{Constraints} 
\end{equation}

\begin{pro}\label{ProDiracBracket}
The restrictions $F_\alpha\big{\vert}_{T^*S^n}=F_\alpha\circ\iota$ 
of the functions (\ref{Falpha}) do Poisson-commute on $T^*S^n$. 
\end{pro}
\begin{proof}
We get, using the Dirac brackets \cite{AM,MR},
\begin{equation}\label{Diracbracket}\begin{array}{l}
\{F_\a\big{\vert}_{T^*S^n},F_\b\big{\vert}_{T^*S^n}\}=
\{F_\a,F_\b\}\big{\vert}_{T^*S^n}\\[6pt]
\displaystyle\hspace{3cm} -
\frac{1}{\{Z_1,Z_2\}}
\left[
\{Z_1,F_\a\}\{Z_2,F_\b\}-\{Z_1,F_\b\}\{Z_2,F_\a\}
\right]\big{\vert}_{T^*S^n}
\end{array}\end{equation}
for second-class constraints. The denominator
$\{Z_1,Z_2\}=-2\sum_{\a=0}^n{q_\a^2}$ does not vanish; one can also check that
$\{Z_1,F_\alpha\}=-4\,a_\alpha p_\alpha q_\alpha(1+Z_1)$ and that 
$\{Z_2,F_\alpha\}=0$ for all $\alpha=0,\ldots,n$. The fact that $\{F_\a,F_\b\}=0$ 
completes the proof.
\end{proof}

\subsubsection{The constrained integrable system as a St\"ackel system}

In order to provide explicit expressions for the sought functions in involution 
$I_1,\ldots,I_n$, we resort to Jacobi 
ellipsoidal coordinates $x^1,\ldots,x^n$ on $S^n$. Those are defined by 
\begin{equation}
Q_\lambda(q,q)=\sum_{\alpha=0}^n{\frac{q_\alpha^2}{a_\alpha-\lambda}}=-\frac{U_x(\lambda)}{V(\lambda)}
\label{EllipsoidalCoordinates}
\end{equation}
where
\begin{equation}
U_x(\lambda)=\prod_{i=1}^n(\lambda-x^i)
\qquad
\mathrm{and}
\qquad
V(\lambda)=\prod_{\alpha=0}^n(\lambda-a_\alpha)
\label{UV}
\end{equation}
and are such that 
\begin{equation}
a_0<x^1<a_1<x^2<\ldots<x^n<a_n.
\label{SeparatingInequalities}
\end{equation}

Notice that Equation (\ref{EllipsoidalCoordinates}) yields the local expressions
\begin{equation}
q_\alpha^2(x)
=
\frac{\displaystyle
\prod_{i=1}^n{(a_\alpha-x^i)}
}{\displaystyle
\prod_{\beta\neq\alpha}{(a_\alpha-a_\beta)}
}
\label{y2}
\end{equation}
for all $\alpha=0,1,\ldots,n$.

\goodbreak

Let us mention the following identity, deduced from (\ref{y2}), viz.,
\begin{equation}
\frac{\partial q_\alpha}{\partial x^i}=-\half\,\frac{q_\alpha}{a_\alpha-x^i}
\label{dy}
\end{equation}
for all $i=1,\ldots,n$ and $\alpha=0,1,\ldots,n$

It is easy to show that the induced metric $\rg_2=(1/B)\sum_{\alpha=0}^n{dq_\alpha^2}\big\vert_{S^n}$, see (\ref{g1g2}), is indeed given by $\rg_2=\sum_{i,j=1}^n{\rg_{ij}(x)dx^idx^j}$ with
\begin{equation}
\rg_{ij}(x)=\frac{1}{4B}\sum_{\alpha=0}^n{
\frac{q_\alpha^2}{(a_\alpha-x^i)(a_\alpha-x^j)}
}
=\rg_i(x)\,\delta_{ij}
\label{gSn}
\end{equation}
where, using a result taken from \cite{Mos}, we have
\begin{equation}
\rg_i(x)=-\frac{1}{4B}\frac{U'_x(x^i)}{V(x^i)}=-\frac{1}{4B}\frac{\prod_{j\neq{}i}(x^i-x^j)}{\prod_\alpha{(x^i-a_\alpha)}}.
\label{MetricEllipsoid}
\end{equation}
This metric is actually positive-definite because of the inequalities (\ref{SeparatingInequalities}). In these ellipsoidal coordinates, we obtain
$$
B=\sum_{\alpha=0}^n{\frac{q_\alpha^2}{a_\alpha}}=\frac{1}{a_0}\frac{x^1\cdots{}x^n}{a_1\cdots{}a_n}.
$$

Upon defining the constrained ``momenta'' $\xi_i$ (for $i=1,\ldots,n$), via the
induced canonical $1$-form 
$\lambda\big{\vert}_{T^*S^n}=\sum_{i=1}^n{\xi_i\,dx^i}=\iota^*\sum_{\alpha=0}^n{p_\alpha{}dq^\alpha}$, we find
\begin{equation}
p_\alpha(\xi,x)=-\frac{q_\alpha(x)}{2B}\sum_{i=1}^n{
\frac{\rg^i(x)\xi_i}{a_\alpha-x^i}.
}
\label{palpha}
\end{equation}

We express, for convenience, the Hamiltonian (\ref{newHam}) on $(T^*S^n,\sum_{i=1}^n{d\xi_i\wedge{}dx^i})$, which is then found to be
\begin{equation}
H=
\frac{1}{2}\sum_{i=1}^n{\rg^i(x)\xi_i^2}
\label{Hsphere}
\end{equation}
where $\rg^i(x)=1/\rg_i(x)$.

Let us now compute the expression of the conserved quantities (\ref{Falpha}) 
on~$T^*S^n$.

\begin{pro}\label{ProDualMoser}
The dual Moser conserved quantities $\left(F_\alpha\big{\vert}_{T^*S^n}\right)_{\alpha=0,\ldots,n}$ 
retain the form
\begin{equation}
F_\alpha\big{\vert}_{T^*S^n}
=
\frac{a_\alpha{}q_\alpha^2(x)}{B}\sum_{i=1}^n{\frac{x^i\rg^i(x)\xi_i^2}{a_\alpha-x^i}}.
\label{FalphaTSn}
\end{equation}
\end{pro}

\begin{proof}
On the one hand, in view of Equations (\ref{y2}) and (\ref{palpha}), one gets, see (\ref{cAcB}), 
$$
\cA_\alpha\big{\vert}_{T^*S^n}=\frac{q_\alpha^2(x)}{B}\sum_{i=1}^n{x^ig^i(x)\xi_i^2},
$$
using the identities
$$
\sum_{\alpha=0}^n{\frac{q_\alpha^2}{a_\alpha-x^i}}=0
$$
for all $i=1,\ldots,n$.

On the other hand, a similar computation gives
$$
\cM_{\alpha\beta}\big{\vert}_{T^*S^n}=\frac{a_\alpha-a_\beta}{2B}q_\alpha(x)q_\beta(x)\sum_{i=1}^n{\frac{x^ig^i(x)\xi_i}{(a_\alpha-x^i)(a_\beta-x^i)}},
$$
so that
$$
\cB_{\alpha}\big{\vert}_{T^*S^n}=
\frac{q_\alpha^2(x)}{B}\sum_{i=1}^n{\frac{(x^i)^2g^i(x)\xi_i^2}{a_\alpha-x^i}},
$$  
proving that $F_{\alpha}\big{\vert}_{T^*S^n}$ (where $F_\alpha=\cA_{\alpha}+\cB_{\alpha}$) is, indeed, as in (\ref{FalphaTSn}).
\end{proof}

\begin{pro} The following holds on $T^*S^n$, viz
\begin{eqnarray}
\label{I0}
\sum_{\alpha=0}^n{F_\alpha\big{\vert}_{T^*S^n}}&=&\frac{1}{B}\sum_{i=1}^n{x^i\rg^i(x)\xi_i^2}=J\big{\vert}_{T^*S^n}\\[6pt]
\sum_{\alpha=0}^n{\frac{F_\alpha\big{\vert}_{T^*S^n}}{a_\alpha}}&=&0\\[6pt]
\sum_{\alpha=0}^n{\frac{F_\alpha\big{\vert}_{T^*S^n}}{a_\alpha^2}}&=&-2H.
\end{eqnarray}
\end{pro}
\begin{proof}
The proof is a direct consequence of Equations (\ref{Rel1}), (\ref{Rel2}), and (\ref{Rel3}), together with the constraints (\ref{Constraints}).
\end{proof}

As a preparation to the proof that our system is, indeed St\"ackel, let us 
introduce, for convenience, the symmetric functions, $\sigma_k(x)$ and $\sigma^i_k(x)$ with 
$x=(x^1,\ldots,x^n)$, that will be useful in the sequel, namely,
\begin{eqnarray}
\label{fctgene2}
U_x(\lambda)\equiv\prod_{j=1}^n(\lambda-x^j)
&=&
\sum_{k=0}^n{(-1)^{k}
\lambda^{n-k}\sigma_{k}(x)
}\\
\label{fctgene}
\frac{U_x(\lambda)}{\lambda-x^i}\equiv\prod_{j\neq{}i}(\lambda-x^j)
&=&
\sum_{k=1}^n{(-1)^{k-1}
\lambda^{n-k}\sigma_{k-1}^i(x)
}.
\end{eqnarray}
We will also use $\sigma_k(a)$ and $\sigma^{\alpha}_k(a)$ with $a=(a_0,a_1,\ldots a_n)$, which are defined similarly.

Let us notice that the previous conserved quantities (\ref{FalphaTSn}) can be written as
$$
F_\alpha\big{\vert}_{T^*S^n}
=
\frac{a_\alpha\,G_{a_\alpha}(\xi,x)}{\displaystyle
\prod_{\beta\neq\alpha}{(a_\alpha-a_\beta)}
}
$$
where
\begin{equation}\label{GenFun}
G_\lambda(\xi,x)
=
\frac{1}{B}\sum_{i=1}^n{
x^i{}\rg^i(x)\prod_{j\neq{}i}{(\lambda-x^j)\,\xi_i^2}.
}
\end{equation}

\goodbreak

\begin{pro}
Let the functions $I_1,\ldots,I_n$ of $T^*S^n$ be defined by
\begin{equation}
G_\lambda(\xi,x)
=
\sum_{k=1}^n{(-1)^{k-1}
\lambda^{n-k}I_k(\xi,x)
}.
\label{GJacobi}
\end{equation}
Then
\begin{equation}
I_k(\xi,x)=\sum_{i=1}^n{
A^i_k(x)\xi_i^2
}
\qquad
\mathrm{with}
\qquad
A^i_k(x)=\frac{1}{B}x^i\rg^i(x)\sigma^i_{k-1}(x).
\label{IkClassicalDualMoser}
\end{equation}
\end{pro}
\begin{proof}
By plugging the definition (\ref{fctgene}) of the symmetric functions $\sigma^i_k(x)$ of order $k=0,1,\ldots,n-1$ (in the variables $(x^1,\ldots,x^n)$, 
with the exclusion of index $i$) into~(\ref{GenFun}), one gets the desired result.
\end{proof}


\goodbreak

\begin{thm}\label{ProJacobi}
The dual Moser system $I_1,\ldots,I_n$ on $T^*S^n$, given by (\ref{IkClassicalDualMoser}), defines a St\"ackel system, with 
St\"{a}ckel matrix $B=A^{-1}$ of the form
\begin{equation}
B^i_k(x^k)=(-1)^i\,\frac{(x^k)^{n-i-1}}{4 V(x^k)}
\label{BN}
\end{equation}
for $i,k=1,\ldots,n$. This implies \cite{Per} that the functions $I_1,\ldots,I_n$ are independent and in involution, which entails that the St\"ackel coordinates $x^1,\ldots,x^n$ are separating for the Hamilton-Jacobi equation.
\end{thm}
\begin{proof}
It is obvious from its expression (\ref{BN}) that $B$ is a St\"{a}ckel matrix \cite{Per}. We just need 
to prove that $A$ is the inverse matrix of $B$. To this aim we first prove a useful identity. 
Let us consider the integral in the complex plane
\[
\frac{1}{2\ii\pi}\int_{\vert{z}\vert=R}\frac{z^{n-i}}{(z-\lambda)}
\frac{U_x(\lambda)}{U_x(z)}\,dz.
\]
When $R\to\infty$ 
the previous integral vanishes because the integrand decreases as $1/R^2$ for large 
$R$ (let us recall that $i\geq1$). We then compute this integral using the theorem of residues and get 
the identity
\begin{equation}\label{usefulidentity}
\sum_{k=1}^n\frac{(x^k)^{n-i}}{U'_x(x^k)}\ \prod_{j\neq k}(\lambda-x^j)=\lambda^{n-i}.
\end{equation}
Equipped with this identity let us now prove that
$\sum_{k=1}^n B^i_k\,A^k_j=\delta^i_j$.
Multiplying this relation by $(-1)^{j-1}\lambda^{n-j}$ and summing over $j$ from $1$ 
to $n$, we get the equivalent relation
\[\sum_{k=1}^n B^i_k\,\sum_{j=1}^n (-1)^{j-1}\lambda^{n-j}A^k_j= 
(-1)^{i-1}\lambda^{n-i},\]
which becomes, using (\ref{IkClassicalDualMoser}) and (\ref{fctgene}):
\[\sum_{k=1}^n B^i_k\,g^k(x)\prod_{j\neq k}(\lambda-x^j)=(-1)^{i-1}\lambda^{n-i}.\]
Using the explicit form of $g^k(x)$ given via (\ref{MetricEllipsoid}) and of the matrix 
$B$, this last relation reduces to the identity (\ref{usefulidentity}), which completes 
the derivation of (\ref{BN}).
\end{proof}


\begin{rmk}\label{potJacobi}
{\rm
A few remarks are in order.
\begin{enumerate}
\item
The first integral defined by (\ref{I0}) is precisely the Joachimsthal invariant (\ref{J}) of the dual Moser system.
\item
It should be emphasized that the bihamiltonian character of our system is obvious with our choice of ellipsoidal coordinates since, from their very definition, the $x^i$ are the eigenvalues of the Benenti $(1,1)$-tensor field $(L_i^j)$ associated with a special conformal Killing tensor. 
\item
One can give some simple potentials for dual Moser. Denoting by $J_k$ the new first integrals, we have
\begin{equation}
J_k=I_k-v_k,\qquad v_k=\mu\sigma_k(x)+\nu(\sigma_1(x)\,\sigma_k(x)-\sigma_{k+1}(x))\end{equation}
with $k=1,\ldots,n$. Those will pairwise Poisson commute (see \cite{Per}, p. 101) if the potential terms can be written in the form
$
v_k=\sum_{i=1}^n\,A^i_k(x)\,f_i(x^i)
$,
implying $f_i(x^i)=\sum_{k=1}^n\,B^k_i\,v_k$. 
A short computation, using the explicit form~(\ref{BN}) of the matrix $B$ and the relation (\ref{fctgene2}), indeed gives
\[f_i(x^i)=\frac{(x^i)^{n-1}}{4V(x^i)}(\mu+\nu\,x^i).\]
\end{enumerate}
}
\end{rmk}

\subsection{Three St\"ackel systems}\label{ThreeStaeckelSystems}

\subsubsection{The Neumann-Uhlenbeck system}

In addition to the two previously studied integrable systems, it may be useful to consider the well-known Neumann-Uhlenbeck system \cite{Uhl1,Uhl2,Mos} on the cotangent bundle of the round sphere $S^n$. It is initially defined on $(T^*\bbR^{n+1},\sum_{\a=0}^n{dp_{\alpha}\wedge dq_{\alpha}})$ by the Hamiltonian
\begin{equation}
H=\half\sum_{\a=0}^n\left(p_\a^2+a_\a q_\a^2\right)
\label{HNeumann}
\end{equation}
with the parameters $0<a_0<a_1<\ldots<a_n.$ 
This system is classically integrable, with the following commuting first integrals 
of the Hamiltonian flow in $T^*\bbR^{n+1}$:
\begin{equation}
F_\alpha(p,q)=q_\alpha^2+\sum_{\beta\neq\alpha}{
\frac{\left(p_\alpha{}q_\beta-p_\beta{}q_\alpha\right)^2}{a_\alpha-a_\beta}
}
\qquad
\mathrm{with}
\qquad
\alpha=0,1,\ldots,n.
\label{IntegralsIninvolutionNeumann}
\end{equation}

Under symplectic reduction, 
with the second class constraints  (\ref{Constraints}), it becomes an integrable system on $(T^*S^n,\sum_{i=1}^n{d\xi_i\wedge{}dx^i})$. Writing $\trg=\sum_{i=1}^n{\trg_i(x)(dx^i)^2}$ the induced Euclidean metric on $S^n$ with $\trg_i(x)=-\frac{1}{4}U'_x(x^i)/V(x^i)$, the independent Poisson-commuting functions $I_k$ ($k=1,\ldots,n$) are 
\begin{equation}
I_k(\xi,x)=\sum_{i=1}^n{
\trg^i(x)\sigma_{k-1}^i(x)\xi_i^2
}
-\s_k(x)
\qquad
\mathrm{with}
\qquad
H=\half{}I_1,
\label{PiNeumann}
\end{equation}
where $\trg^i=1/\trg_i$.

\subsubsection{A synthetic presentation}
Let us observe that the previous calculation enables us to have a synthetic viewpoint unifying the Jacobi-Moser, Neumann-Uhlenbeck, and dual Moser systems. This highlights the novelty of the dual Moser system spelled out in this article.

In Table \ref{Table}, we 
display in each row, and for each system, the metric, the first integrals in involution, 
the Hamiltonian,
and the St\"ackel matrix. (See, e.g., \cite{DV} for a derivation of the formul\ae\ in the first two columns of this table.). Let us emphasize that in all three cases, the metric in the first row is indeed the St\"ackel metric coming from~$I_1$, and which will be, later on, involved in the quantization procedures. 

\begin{table}[h!]
$$
\setlength{\extrarowheight}{8pt}
\begin{array}{|c|c|c|}
\hline
\hbox{Jacobi-Moser} &\hbox{Neumann-Uhlenbeck} & \hbox{dual Moser}\\[10pt]
\hline
\hline
\rg_i=x^i\,\trg_i & \trg_i=-\frac{U'_x(x^i)}{4V(x^i)} & \rg_i=\frac{1}{x^i}\,\trg_i \\[10pt]
\hline
I_k=\sum_i{\frac{\trg^i}{x^i}\,\sigma^i_{k-1}\xi_i^2} & I_k=\sum_i{\trg^i\,\sigma^i_{k-1}\xi_i^2} -\sigma_k(x) & I_k=\sum_i{x^i\,\trg^i\,\sigma^i_{k-1}\xi_i^2} \\[10pt]
\hline
H=\half\,I_1 & H=\half\,I_1& H=\frac{1}{2\sigma_{n+1}(a)}\,I_n \\[10pt]
\hline
x^k\,\tB^i_k(x^k) & \tB^i_k(x^k) = (-1)^i\frac{(x^k)^{n-i}}{4V(x^k)}& \frac{1}{x^k}\,\tB^i_k(x^k) \\[10pt]
\hline
\end{array}
$$
\caption{Three St\"ackel systems} 
\label{Table}
\end{table}

\goodbreak

\section{Quantum integrability}\label{QuantInt}

Start with a configuration manifold $M$ of dimension $n$, and consider the space, $\cS(M)$, of Hamiltonians on $T^*M$ that are fiberwise polynomial. A quantization prescription is a linear iso\-morphism $\cQ$ between this space of \textit{symbols}, $\cS(M)$, and the space, $\cD(M)$, of linear dif\-ferential operators on $M$; this identification is, in addition, assumed to preserve the principal symbol. 

It is well-known that there is, in general, no uniquely defined quantization. However, no matter how the 
quantization is chosen, we will adhere to the following, usual, definition of quantum integrability; see, e.g., \cite{Tot,MT,BCR1,BCR2}.

\goodbreak

\begin{defi}\label{defQuantInt}
A classically integrable system with independent, and mutually Poisson-commuting observables $I_1,\ldots{}I_n$, is integrable 
at the quantum level iff
\begin{equation}
\left[{\cal Q}(I_k),{\cal Q}(I_\ell)\right]=0
\label{QuantumIntegrability}
\end{equation}
for all $k,\ell=1,\ldots,n$
\end{defi}

As a consequence, for a given integrable classical system, depending on the quantization 
procedure used, quantum integrability may be achieved or not. In what follows we will consider and use two quantization schemes for quadratic Hamiltonians: (i) the theory of conformally equivariant quantization, and (ii) Carter's minimal prescription.

\subsection{Conformally equivariant quantization}\label{CEQ}

Let us recall that there exists no quantization mapping that intertwines the action of $\Diff(M)$. To bypass this obstruction, equivariant quantization \cite{LO,DLO} proposes to further endow $M$ with a $G$-structure, and to look under which conditions the existence and uniqueness of a $G$-equivariant quantization can be guaranteed (the proper subgroup $G\subset\Diff(M)$ only is assumed to intertwine the quantization mapping $\cQ$).

We recall that the space $\cF_\l(M)$ of $\l$-densities on $M$, where $\l$ is some complex-valued weight, is the 
space of sections of the complex line bundle 
$\left\vert\Lambda^n{}T^*M\right\vert^\l\otimes\bbC$. 
If $M$ is orientable, $(M,\vol)$, such a 
$\l$-density can be, locally, cast into the form $\phi=f\vert\vol\vert^\l$ 
with $f\in{}C^\infty(M)$; this entails that $\phi$ transforms under the action 
of $a\in\Diff(M)$ according to $f\mapsto{}a_*f\vert(a_*\vol)/\vol\vert^\l$, or infinitesimally as 
\begin{equation}
L^\l_X(f) = 
X(f)+\l\,\Div(X)\,f
\label{Finf}
\end{equation}
for all $X\in\Vect(M)$.

\begin{rmk}\label{RmkHilbert}
{\rm
Note that the completion $\cH(M)$ of the space of compactly supported half-densities, $\cF_\half^c(M)$, 
is a Hilbert space canonically attached to $M$ that will be used in the sequel. The scalar product of two half-densities reads
$$
\la\phi,\psi\ra=\int_M{\!\overline{\phi}\,\psi}
$$
where the bar stands for complex conjugation.
}
\end{rmk}

\goodbreak

We will denote by $\cS_\delta(M)=\cS(M)\otimes\cF_\delta(M)$ the graded space of symbols of weight $\delta$. This space is turned into a $\Vect(M)$-module using the definition (\ref{Finf}) of the Lie derivative extended to the canonical lift of $\Vect(M)$ to $T^*M$.

\goodbreak

Likewise, we will introduce the filtered space $\cD_{\l,\m}(M)$ of differential operators sending $\cF_\l(M)$ to $\cF_\m(M)$. A differential operator of order $k$ is, locally, written as
\begin{equation}
A
=
A_k^{{i_1}\ldots{i_k}}(x)\partial_{i_1}\ldots\partial_{i_k} 
+
\cdots 
+
A_1^i(x)\partial_i 
+
A_0(x)
\label{DiffOp}
\end{equation}
where 
$A_\ell^{{i_1}\ldots{i_\ell}}\in{}C^{\infty}(M)$
for $\ell=0,1,\ldots,k$. It is clear that this space of weighted differential operators, $\cD_{\l,\m}(M)$, becomes a $\Vect(M)$-module via the following definition of the Lie derivative, namely,
\begin{equation}
L^{\l,\m}_X(A) = 
L^{\m}_X\circ{}A - A\circ{}L^{\l}_X
\label{Dinf}
\end{equation}
for all $X\in\Vect(M)$.

From now on, we will be dealing with the case of a conformal (Riemannian) structure, $G=\SO(n+1,1)$, with $n>2$, dictated  by the conformal flatness of our main example: the dual Moser system. 

\begin{thm} [\cite{DLO}]
Given a conformally flat Riemannian manifold $(M,\rg)$, there exists (except for a discrete set of values of $\d=\m-\l$ called resonances) a unique conformally-equivariant quantization, i.e., a linear isomorphism
\begin{equation}
\cQ_{\l,\m}:\cS_\d(M)\to\cD_{\l,\m}(M)
\label{Qlm}
\end{equation}
that (i) preserves the principal symbol, and (ii) intertwines the actions of the Lie algebra $\so(n+1,1)\subset\Vect(M)$.
\end{thm}
In the particular and pivotal case of symbols of degree two, at the core of the present study, explicit formul\ae\ are given by the following theorem. 

\begin{thm}[\cite{DO}]\label{ThmDO}
(i) Let $(M,\rg)$ be a conformally flat Riemannian manifold of dimen\-sion $n\geq3$. The conformally equivariant quantization mapping (\ref{Qlm}) restricted to symbols $P=P_2^{ij}(x)\xi_i\xi_j+P_1^i(x)\xi_i+P_0(x)$ of degree two is given, for non-resonant values of $\delta$, by
$$
\begin{array}{lcl}
\cQ_{\l,\m}(P)  
&=& 
-P_2^{ij}\circ\nabla_i\circ\nabla_j \\[10pt]
&&+ \ii\left(
\beta_1\nabla_iP_2^{ij}+\beta_2\,\rg^{ij}\rg_{k\ell}\nabla_iP_2^{k\ell}
+P_1^j
\right)\circ\nabla_j
\\[10pt]
&& 
+\beta_3\nabla_i\nabla_j(P_2^{ij})
+\beta_4\,\rg^{ij}\rg_{k\ell}\nabla_i\nabla_j(P_2^{k\ell})
+\beta_5R_{ij}P_2^{ij}
+\beta_6R\,\rg_{ij}P_2^{ij}\\[10pt]
&&
+\a\nabla_i(P_1^i)
+P_0\\
\end{array}
$$
where $\nabla$ denotes the levi-Civita connection,\footnote{The 
covariant derivative of $\lambda$-densities $\phi=f\vert\vol_\rg\vert^\l$, locally defined in terms of the Riemannian density, $\vert\vol_\rg\vert$, reads
$\nabla\phi=df\vert\vol_\rg\vert^\l$.} $R_{ij}$ (resp. $R$) the components of the Ricci tensor in the chosen chart (resp. the scalar
curvature) of the metric $\rg$; the coefficients $\a,\beta_1,\ldots,\beta_6$ depend on $\l,\m$, and $n$ in an explicit fashion.\footnote{See Equations (3.3), (3.4), and (4.4) in \cite{DO}.
}

(ii) The quantization mapping $\cQ_{\l,\m}$ depends only on the conformal class of $\rg$.
\end{thm}


The above formula can be specialized to the case of half-density quantization of quadratic symbols $P=P^{ij}(x)\xi_i\xi_j$; one finds\footnote{The value $\delta=0$ is non-resonant \cite{DLO}.}
\begin{equation}
\cQ_{\half,\half}(P)
=
\widehat{P}
+
\beta_3\,\nabla_i\nabla_j(P^{ij})+\beta_4\,\rg^{ij}\rg_{k\ell}\nabla_i\nabla_j(P^{k\ell})
+\beta_5\,R_{ij}P^{ij}+\beta_6\,R\,\rg_{ij}P^{ij}
\label{Qhalfhalf}
\end{equation}
where 
\begin{equation}
\widehat{P}=-\nabla_i\circ{}P^{ij}\circ\nabla_j.
\label{MinQ}
\end{equation}

\begin{rmk}
{\rm
The quantization prescription (\ref{MinQ}), called ``minimal'' in \cite{DV}, has been put forward by Carter \cite{Car}, who dealt with  
polynomial symbols of degree at most two. A great many studies of the quantum spectrum 
for various integrable models use naturally Carter's quantization \cite{Skl,Tot,MT,BT}. 
Along with Equation (\ref{MinQ}), the formul\ae\ for the minimal quantization of lower degree monomials are respectively
\begin{eqnarray}
\widehat{P_0}&=&P_0\\[6pt]
\widehat{P_1}&=&\frac{\ii}{2}\left(P_1^i\circ\nabla_i+\nabla_i\circ P_1^i\right)
\end{eqnarray}
so that 
\beq
\widehat{P_k}=\cQ_{\half,\half}(P_k),
\qquad
\forall k=0,1.
\label{Q0Q1}
\eeq

\goodbreak

Accordingly, a generalization to cubic monomials has been proposed in \cite{DV}:
\beq
\widehat{P_3}=-\frac{\ii}{2}\left(\nabla_i\circ P_3^{ijk}\circ\nabla_j\circ\nabla_k
+\nabla_i\circ\nabla_j\circ P_3^{ijk}\circ\nabla_k\right).
\eeq
All previously defined operators are formally self-adjoint on $\cF_\half^c(M)$; see Remark~\ref{RmkHilbert}.
}
\end{rmk}

In the case where the quadratic observable $P=P^{ij}(x)\xi_i\xi_j$ stems from a Killing tensor,\footnote{This is the case for the integrable systems of St\"ackel type we are studying.} i.e., if $\nabla_{(i}P_{jk)}=0$ for all $i,j,k=1\ldots,n$, we can rewrite Equation (\ref{Qhalfhalf}) as 
\begin{equation}
\cQ_{\half,\half}(P)
=
\widehat{P}
+
f(P)
\label{QuantHalfHalf}
\end{equation}
where $\widehat{P}$ is as in (\ref{MinQ}), and the scalar term is given by
\begin{equation}
f(P)
=
c_1\,\Delta_\rg{}\trP
+c_2\,R_{ij}P^{ij}+c_3\,R\cdot\trP
\label{f(P)}
\end{equation}
where $\Delta_\rg$ is the Laplace operator of $(M,\rg)$, and $\trP=P^{ij}\rg_{ij}$; the coefficients in~(\ref{f(P)}) are respectively
\begin{equation}
c_1= \frac{n^2}{8(n+1)(n+2)},
\quad
c_2= \frac{n^2}{4(n+1)(n-2)},
\quad
c_3= \frac{-n^2}{2(n^2-1)(n^2-4)}.
\label{c1c2c3}
\end{equation}

\subsection{Quantum commutators}\label{Commutators}

In order to implement Definition \ref{defQuantInt} of quantum integrability, we will need some preparation regarding the quantum commutators of Poisson-commuting symbols. In doing so, we will opt for the conformally equivariant quantization $
\cQ\equiv\cQ_{\half,\half}$.

\begin{pro}\label{DualMoserConfEquivQuantInt}
Let $P$ and $Q$ be two, Poisson-commuting, quadratic symbols on $(T^*M,\omega=\sum_{i=1}^n{d\xi_i\wedge{}dx^i})$. The commutator of the two operators $\cQ(P)$ and $\cQ(Q)$, given by 
(\ref{QuantHalfHalf}), retains the form
\begin{equation}
[\cQ(P),\cQ(Q)]
=
\ii\,\cQ(A_{P,Q}+V_{P,Q}),
\label{CommPQ}
\end{equation}
where
\begin{equation}
A_{P,Q}=-\frac{2}{3}\Big(\nabla_j B_{P,Q}^{jk}\Big)\xi_k
\label{Acorr}
\end{equation}
with\footnote{We use the following convention for the Riemann and Ricci tensors, namely, 
$R^{\ell}_{~i,jk}=\partial_j\Gamma^\ell_{ik}+\Gamma^{\ell}_{sj}\Gamma^s_{ik}
-(j\leftrightarrow{}k) $, and $R_{ij}=R^{s}_{~i,sj}$.}
\begin{eqnarray}
B_{P,Q}^{jk}
&=&\nonumber
P^{\ell[j}\nabla_\ell\nabla_m{}Q^{k]m}
+P^{\ell[j}R^{k]}_{~m,n\ell}Q^{mn}-(P\leftrightarrow{}Q)\\[6pt]
&&
-\nabla_\ell{}P^{m[j}\nabla_m{}Q^{k]\ell}-P^{\ell[j}R_{\ell{}m}Q^{k]m}
\label{Btensor}
\end{eqnarray}
and
\begin{equation}
V_{P,Q}=2\left(P^{jk}\,\partial_j f(Q)-Q^{jk}\,\partial_j f(P)\right)\xi_k.
\label{VPQ}
\end{equation}
\end{pro}
\begin{proof}
Start with two quadratic observables $P$ and $Q$. As shown in \cite{DV}, we have $-\ii[\widehat{P},\widehat{Q}]=\widehat{\{P,Q\}}+\widehat{A}_{P,Q}$, where the monomial $A_{P,Q}$, and the skew-symmetric tensor $B_{P,Q}$ are as in (\ref{Acorr}), and (\ref{Btensor}), respectively. If it is then assumed that $\{P,Q\}=0$, Equation (\ref{CommPQ}) follows directly from the explicit expression (\ref{QuantHalfHalf}) of the conformally equivariant quantization mapping, $\cQ$, and from Equation (\ref{Q0Q1}).
\end{proof}

Now, for the Liouville-integrable systems considered below, all Poisson-com\-muting fiberwise polynomial symbols have the form $P=P_2+P_0$, where the indices~$0$ and $2$ refer to the homogeneity degree. In view of Equation (\ref{CommPQ}), and of results obtained in \cite{DV}, we find
$
[\cQ(P_2+P_0),\cQ(Q_2+Q_0)]=\ii\,\cQ(A_{P_2,Q_2}+V_{P_2,Q_2})
$,
which means that the zero degree terms $P_0$ and $Q_0$ produce no quantum cor\-rections.

\goodbreak

The structure of the quantum corrections (the right hand side of Equation~(\ref{CommPQ})) is rather involved, because of the complexity of the 
tensor $B_{P,Q}$; see (\ref{Btensor}). Never\-theless, for St\"ackel systems major simplifications occur.
Indeed, the observables $I_k=I_{2,k}+I_{0,k}$ with $\,I_{2,k}=\sum_i\,g^i(x)\sigma^i_{k-1}(x)\xi_i^2$ 
generate diagonal Killing tensors. Using the separating coordinates $x^i$ and considering
$H=\half{}I_{2,1}=\half\sum_i\,\rg^i(x)\xi_i^2$ for the Hamiltonian fixes up the diagonal metric to be $\rg=\sum_i{\rg_i(x)(dx^i)^2}$, with
$\rg_i=1/\rg^i$ for all $i=1,\ldots,n$. Under these assumptions, Proposition 3.9 in \cite{DV} gives
\beq
B^{k\ell}_{I_{2,i},I_{2,j}}=-2\,I_{2,i}^{s[k}\,R_{st}\,I_{2,j}^{\ell]t}\qquad 
\label{StackelBtensor}
\eeq
for all $i,j,k,\ell=1,\ldots n$, which entails:

\begin{pro}\label{RobertsonCond}
A sufficient condition for a St\"ackel system to be integrable at the quantum level is
\begin{equation}
R_{ij}=0,
\qquad
\forall{}i\neq j
\label{RicDiag}
\end{equation}
where $i,j=1,\ldots,n$, in the special separating coordinate system $(x^i)$.
\end{pro}
\begin{proof}
The Killing tensors $I_{2,i}$ are diagonal, for $i=1,\ldots,n$, in the St\"ackel coordinate system, and the result follows from (\ref{StackelBtensor}). 
\end{proof}

\begin{rmk}
{\rm
\begin{enumerate}
\item Condition (\ref{RicDiag}) is the well-known Robertson condition \cite{Rob}, which has to hold in the separating coordinates system. The relation (\ref{StackelBtensor}) was also obtained in \cite{BCR2} by a direct computation 
of the commutator in separating coordinates; however the explicit form of the tensor 
$B_{P,Q}$ (\ref{Btensor}) was not given there.
\item In Corollary 3.10 of \cite{DV} the Robertson condition was misleadingly claimed to be also necessary.
\item It has been shown by Benenti et al. \cite{BCR1} that the Robertson condition (\ref{RicDiag}) is necessary and sufficient for the separability of the Schr\"odinger equation, comforting the above definition of quantum 
integrability. 
\end{enumerate}
}
\end{rmk}

In the next subsections we will examine, successively, quantum integrability for the following 
St\"ackel systems: the Neumann-Uhlenbeck, the dual Moser and the Jacobi-Moser systems. As previously
explained, the potential, i.e., zero degree terms in the classical observables never induce quantum corrections; they will therefore be systematical\-ly omitted.

\subsection{The quantum Neumann-Uhlenbeck system}\label{QuantNeumann}

Let us recall that, for the Neumann-Uhlenbeck system (see Table \ref{Table}), the St\"ackel metric $\trg=\sum_i{\trg_i(x)(dx^i)^2}$, is
\begin{equation}
\trg_i(x)=-\frac{1}{4}\frac{U'_x(x^i)}{V(x^i)}=-\frac{1}{4}\frac{\prod_{j\neq{}i}(x^i-x^j)}{\prod_\alpha{(x^i-a_\alpha)}}
\label{MetricNeumann}
\end{equation}
for $i=1,\ldots,n$. If we put $\trg^i=1/\trg_i$, the independent and Poisson-commuting observables are given by
$$
I_{k}=\sum_{i=1}^n{\trg^i(x)\sigma^i_{k-1}(x)\,\xi_i^2}
$$
for $k=1,2,\ldots,n$, and the St\"ackel Hamiltonian is $H=\half{}I_1=\half\sum_i{\trg^i(x)\,\xi_i^2}$.
\begin{pro}\label{ProNeumannUhlenbeck}
The conformally equivariant quantization does preserve quantum integrability of 
the Neumann-Uhlenbeck system.
\end{pro}
\begin{proof}
From the fact that $(S^n,\trg)$ is the round sphere, we have 
\begin{equation}
\widetilde{R}_{ij}=(n-1)\,\trg_i\,\delta_{ij}
\qquad 
\& 
\qquad \widetilde{R}=n(n-1).
\label{Ricci+R-Neumann}
\end{equation}
Straightforward computation then leads to
\beq
f(I_k)=(n-k+1)[c_4\,\sigma_{k-1}(x)+2(n-k+2)c_1\,\sigma_{k-1}(a)]
\label{fkNeumann}
\eeq
for all $k=1,\ldots,n$,
where
\begin{equation}
c_4=-2(n+1)c_1+(n-1)c_2+n(n-1)c_3.
\label{c4}
\end{equation}
Relations (\ref{c1c2c3}) readily imply the vanishing of $c_4$. As a consequence, the $f(I_k)$ are just constant, ensuring that $V_{I_k,I_\ell}=0$ (see (\ref{VPQ})). Equation (\ref{CommPQ}) and the fact that $B_{I_k,I_\ell}=0$ (since the Ricci tensor is diagonal in this coordinate system) entail that the conformally equivariant quantization (which coincides, up to a constant term, with Carter's) preserves integrability of the system at the quantum level.                 
\end{proof}

\subsection{The quantum dual Moser system}\label{QuantDualMoser}

In the basic geometrical construction of the Poisson-commuting conserved quantities $I_k$, we have been considering the conformally flat metric $\rg_2$ given by (\ref{gSn}) and~(\ref{MetricEllipsoid}). Now, in the quantum approach to integrability, we 
choose to use, again, 
the St\"ackel metric, $\rg$, associated with $I_1$. One has
\begin{equation}
\rg=\sum_{i=1}^n{\rg_i(x)(dx^i)^2},\qquad\mbox{with}\qquad \rg_i(x)=\frac{1}{x^i}\,\trg_i(x),
\label{MetricDualMoser}
\end{equation}
where the Neumann-Uhlenbeck metric $\trg$ is given by (\ref{MetricNeumann}), while the first integrals for $k=1,\ldots,n$ are
\begin{equation}
I_k=\sum_{i=1}^n\rg^i(x)\,\sigma^i_{k-1}(x)\xi_i^2,\qquad \rg^i(x)=\frac 1{\rg_i(x)}.
\label{IkDualMoser}
\end{equation}

\begin{lem}\label{LemmaDualMoser}
The metric (\ref{MetricDualMoser}) has Ricci tensor 
\begin{equation}
R_{ij}=\Big((n-2)x^i+n\sum_{k=1}^n{x^k}-(n-1)\sum_{\alpha=0}^n{a_\alpha}\Big)\rg_i\,\delta_{ij}
\label{RicciDual-Moser}
\end{equation}
and scalar curvature
\begin{equation}
R=(n-1)\Big((n+2)\sum_{k=1}^n{x^k}-n\sum_{\alpha=0}^n{a_\alpha}\Big).
\label{R-Moser}
\end{equation}
It is conformally flat for $n=\dim(M)\geq 3$.
\end{lem}

\begin{proof}
The Ricci tensor can be computed with the help of classical formul\ae\ for a diagonal metric 
(see for instance \cite{LPEbook}, p. 119). The only possibly non-vanishing components  of the 
Riemann tensor are $R_{ik,kj}$, for $i\neq{}j\neq{}k$, and $R_{ij,ji}$, for $i\neq{}j$.
Using the relations
$$
\partial_i(\ln\rg_j)=\frac 1{x^i-x^j}\quad (i\neq{}j),
\qquad
\quad 
\partial_{ij}(\ln\rg_k)=0\quad (i\neq{}j\neq{}k)
$$
one easily gets $R_{ik,kj}=0$, implying
$$
R_{ij}=-\sum_{k=1}^n
{\rg^k{}R_{ik,kj}}=0,
\qquad 
\forall{}i\neq j.
$$
The computation of the remaining components involves a sum which is conveniently computed 
using the theorem of residues, giving
$$
R_{ik,ik}=(x^i+x^k+\sum_{s=1}^n{x^s}-\sum_{\alpha=0}^n{a_\alpha})\,\rg_i\rg_k
$$
from which one deduces easily the diagonal part of the Ricci tensor, given by (\ref{RicciDual-Moser}), and the scalar curvature (\ref{R-Moser}).
Some extra computation shows that the conformal Weyl tensor vanishes in dimension $n\geq 4$, and 
that the Cotton-York tensor vanishes as well for $n=3$.
\end{proof}
\begin{rmk}
{\rm
Although the metric $\rg_2$ given by (\ref{gSn}) is clearly conformally flat, it is by no means trivial that the same is true for the St\"ackel metric, $\rg$, given by (\ref{MetricDualMoser}) on $S^n$.
}
\end{rmk}
We are now in position to prove the following proposition.
\begin{pro}\label{ProConfQ}
The conformally equivariant quantization procedure (\ref{QuantHalfHalf}) does preserve quantum integrability of 
the dual Moser system.
\end{pro}
\begin{proof}

Using the definition (\ref{f(P)}) of the scalar term in the formula (\ref{QuantHalfHalf}) for the conformally equivariant quantization of the~$I_k$, we find
\begin{eqnarray*}
f(I_k)&=&[(n-2)c_2+k(c_6-c_4)]\sigma_1(x)\,\sigma_k(x)+[kc_5-(n-2)c_2]\sigma_{k+1}(x)\\[4mm]
&&-kc_4\sigma_1(a)\,\sigma_k(x)-2k(n-k+1)c_1\,\sigma_{k+1}(a)
\end{eqnarray*}
where $c_4$ was already defined in (\ref{c4}) and shown to vanish in the proof of Proposition~\ref{ProNeumannUhlenbeck}; we also have
\[c_5=2(n+2)c_1-(n-2)c_2,\qquad\qquad c_6=-2c_1+c_2+2(n-1)c_3.\]
Taking into account the relations (\ref{c1c2c3}) one gets $c_5=c_6=0$, and we are left, 
for $\,k=1,\ldots,n$, with
\begin{equation}
f(I_k)=
2c_1\Big[(n+2)\left[\sigma_k(x)\sigma_1(x)-\sigma_{k+1}(x)\right]-k(n-k+1)\,\sigma_{k+1}(a)\Big]
\label{fIk-Dual-Moser}
\end{equation}
where we posit $\sigma_{n+1}(x)=0$.

Let us now compute $V_{I_k,I_\ell}$ defined by (\ref{VPQ}) needed to check quantum integrability via the commutator (\ref{CommPQ}). 

In view of (\ref{IkDualMoser}), one finds
$$
V_{I_k,I_\ell}=2\sum_{i=1}^n
{\rg^i\Big(\sigma^i_{k-1}\partial_i f(I_\ell)
-\sigma^i_{\ell-1}\partial_i f(I_k)\Big)\xi_i
}.
$$
Now, using the relations \cite{Ben}
$$
\partial_i\sigma_k(x)=\sigma^i_{k-1}(x),\qquad 
\forall i,k=1,\ldots,n,
$$
and
$$
\sigma_k(x)=\sigma^i_k(x)+x^i\sigma^i_{k-1}(x),\qquad 
\forall
k=1,\ldots,n-1,
$$
as well as
$$
\sigma_n(x)=x^i\,\sigma^i_{n-1}(x),
$$
one gets
$
\partial_i f(I_k)=2(n+2)c_1[x^i+\sigma_1(x)]\sigma^i_{k-1}(x)
$
which obviously yields
$V_{I_k,I_\ell}=0$, implying, at last
\begin{equation}
\left[\cQ(I_k),\cQ(I_\ell)\right]=0
\label{QIkIlDual-Moser=0}
\end{equation}
for all $k,\ell=1,\ldots,n$.
\end{proof}

\begin{rmk}
{\rm
Carter's (minimal) prescription (\ref{MinQ}) also leads to quantum integra\-bility of the system because of the diagonal form (\ref{RicciDual-Moser}) of the Ricci tensor in the separating coordinates. Now, in contradistinction with the Neumann-Uhlenbeck quantum system, the scalar terms $f(I_k)$ given by (\ref{fIk-Dual-Moser}) are no longer constant, yielding quite different quantum observables $\cQ(I_k)$ and $\hat{I_k}$. So, the fact that quantum integrability is  
not only preserved by Carter's quantum prescription, but also by conformally equivariant quantization is a new and noteworthy phenomenon.
} 
\end{rmk}

\subsection{The quantum Jacobi-Moser system}\label{QuantJacobiMoser}
The St\"ackel metric, associated to $I_1$ is now (see Table \ref{Table}):
\begin{equation}
\rg=\sum_{i=1}^n{\rg_i(x)(dx^i)^2},\qquad\mbox{with}\qquad \rg_i(x)=x^i\,\trg_i(x),
\label{MetricJacobiMoser}
\end{equation}
where the Neumann-Uhlenbeck metric $\trg_i(x)$ is given by (\ref{MetricNeumann}), and the first integrals by
\begin{equation}
I_k=\sum_{i=1}^n{\rg^i(x)\sigma^i_{k-1}(x)\xi_i^2},
\qquad\quad 
\rg^i(x)=\frac{1}{\rg_i(x)}.
\label{IkJacobiMoser}
\end{equation}
 
\begin{lem}
The Ricci tensor and the scalar curvature of the metric (\ref{MetricJacobiMoser}) are given by
\begin{equation}
R_{ij}=\frac{\sigma_{n+1}(a)}{\sigma_n^2(x)}\,\sigma^i_{n-2}(x)\,\rg_i\,\delta_{ij},\qquad
R=2\,\frac{\sigma_{n+1}(a)}{\sigma_n^2(x)}\,\sigma_{n-2}(x).
\label{Ricci-RJacobiMoser}
\end{equation}
\end{lem}
\begin{proof} It is completely similar to the proof of Lemma \ref{LemmaDualMoser}.
\end{proof}

This allows us to prove:
\begin{pro}
Carter's prescription (\ref{MinQ}) preserves quantum integrability of the Jacobi-Moser system, while the conformally equivariant quantization does not.
\end{pro}
\begin{proof}
Since the Ricci tensor is diagonal, quantum integrability is established for the 
prescription (\ref{MinQ}). 
For the conformally equivariant quantization (\ref{QuantHalfHalf}), we will just give a counter-example. One has
\[
f(I_1)=2(c_2+nc_3)\,\sigma_{n+1}(a)\,\frac{\sigma_{n-2}(x)}{\sigma^2_n(x)}
\]
and
\[
f(I_2)=(n-1)c_3\,\frac{\sigma_{n+1}(a)}{\sigma^2_n(x)}\Big[2(n-1)\,\sigma_{n-1}(x)
-n\,\sigma_1(x)\,\sigma_{n-2}(x)\Big]-2n(n-1)\,c_1.
\]
A simple computation gives
$V^i_{I_1,I_2}=\partial{V_{I_1,I_2}}/{\partial\xi_i}=\partial_i f(I_2)-(\sigma_1(x)-x^i)\partial_i f(I_1)$,
hence the non-vanishing result
\begin{eqnarray*}
\displaystyle 
V^i_{I_1,I_2}&=&(-c_3)\frac{\sigma_{n+1}(a)}{x^i\,\sigma^2_n(x)}\Big[-2(\sigma_1^i(x)\sigma_{n-2}(x)+
\sigma_1(x)\sigma^i_{n-2}(x))\\[4mm]
&&
+(n^2-3n+4)\sigma_{n-1}(x)+n(3n-5)\sigma^i_{n-1}(x)\Big],
\end{eqnarray*}
showing that the system looses its quantum integrability via conformally equivariant quantization.
\end{proof}

\begin{rmk}
{\rm
Let us mention that quantum integrability of the Neumann-Uhlen\-beck and Jacobi-Moser systems has first been established, in terms of Carter's quantum prescription~(\ref{MinQ}), by Toth \cite{Tot}.
}
\end{rmk}

\section{Conclusion and outlook}\label{Conclusion}

To sum up the main results of the article, let us mention that we have disclosed a new integrable system on $S^n$, in duality with the well-known Jacobi-Moser system in terms of projective equivalence. As opposed to that of the generic ellipsoid, the ``dual'' metric is conformally flat. This remarkable fact enables us to have naturally recourse to conformally equivariant quantization. The latter turns out to preserve integrability at the quantum level. It is, to our knowledge, the first instance of conformally driven quantum integrability.

This opens new perspectives related, e.g., to the determination of the conditions  under which a classically integrable system, stemming from second-order Killing tensors on a conformal\-ly flat configuration manifold, remains quantum-integrable via the conformally equi\-variant quantization. Also, possible generalizations of the Jacobi-Moser system and its dual counterpart might conceivably be put to light in a similar manner.

\bigskip

\textbf{Acknowledgements}: We express our deep gratitude to Serge Tabachnikov and Valentin Ovsienko for helpful discussions which have triggered this work. We are also indebted to John Harnad for useful correspondence. 

\goodbreak



\begin{thebibliography}{99}

\bibitem{AM}
R. Abraham, and J.E. Marsden,
\textsl{Foundations of Mechanics},
Second Edition,
Addison-Wesley Publishing Company, Inc (1987).


\bibitem{BT}
M. Bellon, M. Talon,
{\it The quantum Neumann model: refined semiclassical results}, Phys. Lett. A, {\bf 356} 
(2006) 110--114, and references therein.

\bibitem{Ben}
S. Benenti,
{\it Inertia tensors and St\"ackel systems in the Euclidean spaces}, Rend. Sem. 
Mat. Univ. Pol. Torino {\bf 50} (1992) 315--341.

\bibitem{BCR1}
S. Benenti, C. Chanu, and G. Rastelli,
{\it Remarks on the connection between the additive separation of the Hamilton-Jacobi 
equation and the multiplicative separation of the Schr\"odinger equation. I. The 
completeness and Robertson conditions}, J. Math. Phys. {\bf 43} (2002) 5183--5222.

\bibitem{BCR2}
S. Benenti, C. Chanu, and G. Rastelli,
{\it Remarks on the connection between the additive separation of the Hamilton-Jacobi 
equation and the multiplicative separation of the Schr\"odinger equation. II. First 
integrals and symmetry operators}, J. Math. Phys. {\bf 43} (2002) 5223--5253.

\bibitem{BM}
A.V. Bolsinov, V.S. Matveev, 
{\it Geometrical interpretation of Benenti systems}, 
J. Geom. Phys., {\bf 44}:4 (2003) 489--506.
 	
\bibitem{Car}
B. Carter,
{\it Killing tensor quantum numbers and conserved currents in curved space}, 
Phys. Rev. D{\bf 16} (1977) 3395--3414.





\bibitem{DV}
C. Duval, G. Valent,
{\it Quantum Integrability of Quadratic Killing Tensors},
J. Math. Phys. \textbf{46} (2005) 053516(22).

\bibitem{DO}
C. Duval, V. Ovsienko,
{\it Conformally equivariant quantum Hamiltonians}, 
Selecta Math. (N.S.) {\bf 7}:3 (2001) 291--320.


\bibitem{DLO}
C. Duval, P. Lecomte, V. Ovsienko,
{\it Conformally equivariant quantization: existence and uniqueness},
Ann. Inst. Fourier. {\bf  49}:6 (1999) 1999--2029.


\bibitem{LPEbook}
L. P. Eisenhart,
{\sl Riemannian geometry}, Princeton Landmarks in Mathematics and Physics, 
Princeton University Press (1997).

\bibitem{FP}
G. Falqui, M. Pedroni,
{\it Separation of Variables for Bi-Hamiltonian Systems}, 
Math. Phys. Anal. Geom. {\bf 6} (2003) 139--179.


\bibitem{IMM}
A. Ibort, F. Magri, G. Marmo,
{\it Bihamiltonian structures and St\"ackel separability},
J. Geom. Phys. {\bf 33} (2000) 210--228.


\bibitem{Kun}
H.P.~K\"unzle, 
``Galilei and Lorentz structures on space\-time:
Comparison of the corresponding geometry and physics,'' 
Ann. Inst. H. Poincar\'e, 
Phys. Th\'eor {\bf 17} (1972),  337--362.

\bibitem{LO}
P.B.A. Lecomte and V. Ovsienko,
{\it Projectively invariant symbol calculus}, Lett. Math. Phys. 
{\bf 49}:3 (1999) 173--196.

\bibitem{LevCiv}
T. Levi-Civita,
{\it Sulle trasformazioni delle equazioni dinamiche}, Ann. Mat. Ser. 2a 
{\bf 24} (1896) 255--300.


\bibitem{Mag}
F. Magri, 
{\it A simple model of the integrable Hamiltonian equation},
J. Math. Phys. {\bf 19}:5 (1978) 1156--1162.


\bibitem{MT}
V.S.~Matveev, P.J.~Topalov,
{\it Quantum integrability of Beltrami-Laplace operator as geodesic equivalence}, Math. Z. {\bf 238} (2001) 833--866.

\bibitem{Mos}
J. Moser,
{\it Various aspects of the integrable Hamiltonian systems}, in Dynamical 
systems (C.I.M.E. Summer School Bressanone; 1978), Progress in Mathematics, 
Birkh\"auser {\bf 8} (1981) 233--289.


\bibitem{OT}
V. Ovsienko, and S. Tabachnikov, 
{\em Projective Differential Geometry Old And New: From The Schwarzian Derivative To The Cohomology Of Diffeomorphism Group}, 
Cambridge University Press (2005).

\bibitem{MR}
J.-E. Marsden, T. Ratiu, 
{\em Introduction to Mechanics and Symmetry}, Springer, 1999.

\bibitem{Per}
A.M. Perelomov,
{\em Integrable Systems of Classical Mechanics and Lie Algebras}, Vol I,
Birkh\"auser (1990), and references therein.

\bibitem{Rob}
H. P. Robertson, 
{\it Bermerkung \"uber separierbare Systeme in der Wellenmechanik},
Math. Annal. {\bf 98} (1927) 749--752.

\bibitem{Skl}
E.K. Sklyanin,
{\it The quantum Toda chain},
in Non-linear equations in classical
and quantum field theory. Springer Notes in Physics, vol. 226, (1985).


\bibitem{Tab1}
S. Tabachnikov, 
{\it Projectively equivalent metrics, exact transverse line fields and the geodesic flow on the ellipsoid},
Comm. Math. Helv. {\bf 74}:1 (1999) 306--321.

\bibitem{Tab2}
S. Tabachnikov, 
{\it Ellipsoids, complete integrability and hyperbolic geometry},
Moscow Mathematical Journal,
{\bf 2}:1, (2002) 185--198.

\bibitem{TM}
P.J. Topalov, V.S. Matveev,
{\it Geodesic equivalence and integrability},
MPIM preprint series, no. 74 (1998), 
\texttt{arXiv:math/9911062v1 [math.DG]}.

\bibitem{Tot}
J.A. Toth, 
{\it Various quantum mechanical aspects of quadratic forms},
J. Funct. Anal. {\bf 130}:1 (1995) 1--42.

\bibitem{Uhl1}
K. Uhlenbeck,
{\it Minimal $2$-spheres and tori in $S^k$},
Preprint University of Illinois at Chicago Circle (1975).

\bibitem{Uhl2}
K. Uhlenbeck,
{\it Equivariant harmonic maps into spheres},
in Proceedings of the 1980 NSF-CBMS Regional Conference on ``Harmonic Maps'', R.J.~Knill, M.~Kalka, and H.C.J.~Sealey (Eds), LNM 949 (1982) 146--158.


\end{thebibliography}
\end{document}